\documentclass[journal,12pt,draftcls,onecolumn,a4paper]{IEEEtran}
\IEEEoverridecommandlockouts
\usepackage{graphicx}
\usepackage{amsmath,bbm,epsfig,amssymb,amsfonts, amstext, verbatim,amsopn,cite,subfigure,multirow,multicol,lipsum}
\usepackage{balance}
\usepackage{booktabs,threeparttable}
\usepackage{amsfonts}
\usepackage{enumerate}   
\newcommand{\R}{\mathbb{R}}
\usepackage{epsfig}
\usepackage{epstopdf}
\usepackage{setspace}
\usepackage{stmaryrd}
\usepackage{psfrag}	
\usepackage{multirow}
\usepackage{float}
\usepackage[process=auto]{pstool}
\usepackage{etoolbox}
\usepackage{algorithm}
\usepackage{algorithmic}
\usepackage{hyperref}
\usepackage{pifont}
\allowdisplaybreaks
\usepackage{amsmath,amssymb,amsfonts}
\usepackage{algorithmic}
\usepackage{graphicx}
\usepackage{textcomp}
\usepackage{bm}
\usepackage{amsmath,bbm,amssymb,amsfonts,amstext,amsopn}
\usepackage{amssymb}

\usepackage{amsthm}

\usepackage{cite}
\usepackage{balance}
\allowbreak
\usepackage{url}
\usepackage{epsfig}
\usepackage{setspace}
\usepackage{stmaryrd}
\usepackage{psfrag}	
\usepackage{float}
\usepackage{amsthm}

\newcommand{\subparagraph}{}
\usepackage[compact]{titlesec}
\usepackage[english]{babel}
\usepackage[process=auto]{pstool}
\usepackage{amsthm}

\theoremstyle{remark}
\newtheorem{remark}{Remark}
\usepackage{epstopdf}
\usepackage{algorithmic}
\usepackage{xcolor}
\usepackage{etoolbox}
\usepackage{algorithm}
\usepackage[process=auto]{pstool}
\usepackage{epsfig}

\newtheorem{lem}{Lemma}

\newtheorem{defn}{Definition}

\begin{document}
	\title{Optimal Resource Allocation for Multi-user OFDMA-URLLC MEC Systems}
\author{\IEEEauthorblockN{Walid R. Ghanem,Vahid Jamali, and Robert Schober}
	\thanks{This paper will be presented in part at IEEE GLOBECOM 2020 \cite{gha4}.}
	\thanks{The authors are with the Institute for Digital Communications, Friedrich-Alexander-University Erlangen-N\"urnberg (FAU), Germany (email: \{walid.ghanem, vahid.jamali, and robert.schober\}@fau.de).}}
\maketitle
\vspace{-2cm}
\begin{abstract}
In this paper, we study resource allocation algorithm design for multi-user orthogonal frequency division multiple access (OFDMA) ultra-reliable low latency communication (URLLC) in mobile edge computing (MEC) systems. To meet the stringent end-to-end delay and reliability requirements of URLLC MEC systems, we propose joint uplink-downlink resource allocation and finite blocklength transmission. Furthermore, we employ a partial time overlap between the uplink and downlink frames to minimize the end-to-end delay, which introduces a new time causality constraint. The proposed resource allocation algorithm is formulated as an optimization problem for minimization of the total weighted power consumption of the network under a constraint on the number of URLLC user bits computed within the maximum allowable computation time, i.e., the end-to-end delay of a computation task.  Despite the non-convexity of the formulated optimization problem, we develop a globally optimal solution using  a branch-and-bound approach based on discrete monotonic optimization theory. The branch-and-bound algorithm minimizes an upper bound on the total power consumption until convergence to the globally optimal value. Furthermore, to strike a balance between computational complexity and performance, we propose two efficient suboptimal algorithms based on successive convex approximation and second-order cone techniques. Our simulation results reveal that the proposed resource allocation algorithm design facilitates URLLC in MEC systems, and yields significant power savings compared to three baseline schemes. Moreover, our simulation results show that the proposed suboptimal algorithms offer different trade-offs between performance and complexity and attain a close-to-optimal performance at comparatively low complexity.
\end{abstract}
\section{Introduction} 
Future wireless communication networks target several objectives including high data rates, reduced latency, and massive device connectivity. One important objective is to facilitate ultra-reliable low latency communication (URLLC). URLLC is crucial for mission-critical applications such as remote surgery, factory automation, autonomous driving, tactile Internet, and augmented reality to enable real-time machine-to-machine and human-to-machine interaction \cite{popski}. URLLC imposes strict quality-of-service (QoS) constraints including a very low latency (e.g., 1 $\textrm{ms}$) and a low packet error probability (e.g., $10^{-6}$). 

Recently, significant attention has been devoted to studying and developing resource allocation algorithms for URLLC. In particular, optimal power allocation in a multi-user time division multiple access (TDMA) URLLC system was considered in \cite{optimal,convexfinite}. Moreover, resource allocation for orthogonal frequency division multiple access (OFDMA)-URLLC systems was studied in \cite{chsejoint,ghanem1,gha,cellurllc,darbi1}. In \cite{Rensecurrlc,Ghan3}, resource allocation for secure URLLC was investigated. However, the resource allocation schemes in \cite{optimal,convexfinite,ghanem1,gha,cellurllc,darbi1,Rensecurrlc,Ghan3} focused only on communication while computation was not considered. Nevertheless, devices in mission-critical applications are expected to  generate tasks that require computation within a given time. This motivates the investigation of resource allocation algorithm design for efficient computation in URLLC systems.

A promising solution to enable efficient and fast computation for URLLC devices is mobile edge computing (MEC). MEC can enhance the battery lifetime and reduces the power consumption of users with delay-sensitive computation tasks\cite{dynamic}. By offloading these tasks to nearby MEC servers, the power consumption and computation time at the local users can be considerably reduced at the expense of the power required for data transmission for offloading\cite{dynamic}. Thus, careful resource allocation is paramount for MEC to ensure the efficient use of the available resources (e.g., power and bandwidth) while guaranteeing a maximum delay for the computation tasks. Existing resource allocation algorithms for MEC systems, such as \cite{energymec,Zhoumec,edge,ata5a}, are based on Shannon's capacity formula. In particular,  the authors of \cite{energymec,edge} studied energy-efficient resource allocation for MEC, while computation rate maximization was targeted in \cite{Zhoumec}. However, if the resource allocation design for URLLC MEC systems is based on Shannon's capacity formula, the reliability of the offloading and downloading processes cannot be guaranteed because of the imposed delay constraints. To overcome this issue, recent works applied finite blocklength transmission (FBT)\cite{Polyanskiy} for resource allocation algorithm design for URLLC MEC systems. In particular, the authors in \cite{9048917} studied binary offloading in single-carrier TDMA systems. However, single-carrier systems suffer from poor spectrum utilization and require complex equalization at the receiver. In \cite{DeepMEC}, the authors investigated the minimization of the normalized energy consumption of an OFDMA-URLLC MEC systems. However, the algorithm proposed in \cite{DeepMEC} assumes that the channel gains of different sub-carriers are identical which may not be a realistic assumption for broadband wireless channels. Moreover, the resource allocation algorithms proposed in \cite{DeepMEC} are based on a simplified version of the general expression for the achievable rate for FBT \cite{Polyanskiy}. Furthermore, the existing MEC designs, such as \cite{energymec,uavmec}, do not take into account the size of the computation result of the tasks and do not consider the communication resources consumed for downloading of the processed data by the users. Nevertheless, the size of the processed data can be large for applications such as augmented reality.

 We note that most resource allocation algorithms proposed for URLLC systems in the literature, such as \cite{9088229,Rensecurrlc,darbi1,ghanem1}, are strictly suboptimal. In particular, the algorithms developed in \cite{9088229,Rensecurrlc} were based on block coordinate descent techniques, while those in \cite{darbi1,ghanem1} employed successive convex approximation (SCA). As a result, the performance of the resource allocation algorithms in \cite{9088229,Rensecurrlc,darbi1,ghanem1} cannot be guaranteed because the gap between the optimal and  suboptimal solutions is not known. To cope with this problem, in our recent work \cite{gha}, we proposed a global optimal algorithm based on the polyblock outer approximation method using monotonic optimization. However, the polyblock algorithm may suffer from  slow convergence for large problem sizes. To overcome this problem, in this paper, a branch-and-bound algorithm is proposed. Different from the general branch-and-bound algorithms proposed for non-convex problems, e.g., \cite{branch}, the proposed branch-and-bound algorithm exploits the monotonicity of the problem to reduce the search space for faster convergence\cite{biofast}.
 
In this paper, we study optimal joint uplink-downlink resource allocation for OFDMA-URLLC MEC systems. The main contributions of this paper are as follows:    
\begin{itemize}
	\item We propose a novel joint uplink-downlink resource allocation algorithm design for multi-user OFDMA-URLLC MEC systems. To reduce the end-to-end delay of uplink and downlink transmission while efficiently exploiting the available spectrum, we propose a partial time overlap between the uplink and downlink frames and introduce corresponding causality constraints. Then, the resource allocation algorithm design is formulated as an optimization problem for the minimization of the total weighted power consumed by the base station (BS) and the users subject to QoS constraints for the URLLC users. The QoS constraints include the required number of bits computed within a maximum allowable time, i.e., the maximum end-to-end delay of the users. 
	\item  The formulated optimization problem is a non-convex mixed-integer problem which is difficult to solve. Thus, we transform the problem into the canonical form of a discrete monotonic optimization problem. This reformulation allows the application of the branch-and-bound algorithm to find the global optimal solution. The proposed branch-and-bound algorithm searches for a global optimal solution by successively partitioning the non-convex feasible region and  using bounds on the objective function to discard inferior partition elements.
	\item To strike a balance between computational complexity and performance, we develop two efficient low-complexity suboptimal algorithms based on SCA and second-order cone programming (SOC).
	\item Our simulations show that the proposed suboptimal algorithms offer different trade-offs between complexity and performance and closely approach the performance of the optimal algorithm, despite their significantly lower complexity.  Furthermore, the proposed algorithms achieve significant performance gains compared to three baseline schemes.
\end{itemize}

We note that this paper expands the corresponding conference version \cite{gha4} in several directions. First, the formulated optimization problem targets joint local computing and edge offloading, while only  edge offloading was considered in \cite{gha4}. Second, we derive the \textit{optimal} resource allocation policy for OFDMA-URLLC MEC systems, whereas only a suboptimal algorithm was provided in [1]. Thirdly, we propose a second suboptimal algorithm to further reduce the complexity of the  suboptimal scheme proposed in \cite{gha4}. 

The remainder of this paper is organized as follows. In
Section II, we present the considered system and channel models. In Section III, the proposed resource allocation problem is formulated. In Section IV, the optimal resource allocation algorithm is derived, whereas low-complexity suboptimal algorithms are provided in Section V. In Section VI, the performance of the proposed schemes is evaluated via computer simulations, and finally conclusions are drawn in Section VII.

\textit{Notation}: Lower-case letters $x$ refer to scalar numbers, and bold lower-case letters $\mathbf{x}$ represent vectors. $(\cdot)^{T}$ denotes the transpose operator. 
$\mathbb{R}^{N \times 1}$ represents the set of all $N \times 1$ vectors with real valued entries. The circularly symmetric complex Gaussian distribution with mean $\mu$ and variance $\sigma^{2}$ is denoted by $\mathcal{CN}(\mu,\sigma^{2})$, $\sim$ stands for ``distributed as", and $\mathcal{E}\{\cdot\}$ denotes statistical expectation. $\nabla_{\mathbf{x}}f(\mathbf{x})$ denotes the gradient vector of function $f(\mathbf{x})$ and its elements are the partial derivatives of $f(\mathbf{x})$. For any two vectors $\mathbf{x}$, $\mathbf{y}$ $\in$ $\R_{+}$, $\mathbf{x} \leq \mathbf{y}$ means $x_{i}\leq y_{i}$, $\forall i,$ where $x_{i}$ and $y_{i}$ are the $i$-th elements of $\mathbf{x}$ and $\mathbf{y}$, respectively. $\mathbf{x}^{*}$ denotes the optimal value of an optimization variable $\mathbf{x}$.\vspace{-0.2cm}
\section{System and Channel Models}
In this section, we present the system and channel models for the considered OFDMA-URLLC MEC system.
\subsection{System Model}
We consider a single-cell multi-user MEC system which comprises a BS and $K$ URLLC users indexed by $k =\{1,\dots,K\}$, cf. Fig.~\ref{model}. All transceivers have single antennas. The system employs frequency division duplex (FDD)\footnote{In FDD systems, different frequency bands are assigned to uplink and downlink.}. Thereby, the total bandwidth $W$ is divided into two bands for uplink and downlink transmission having bandwidths $W^{u}$ and $W^{d}$, respectively. The bandwidths for uplink and downlink transmission are further divided into $M^{u}$ and $M^{d}$ orthogonal sub-carriers indexed by $m^{u} =\{1,\dots,M^{u}\}$ and $m^{d} =\{1,\dots,M^{d}\}$, respectively. The bandwidth of each sub-carrier is $B_{s}$, leading to a symbol duration of $T_{s}=\frac{1}{B_{s}}$. The uplink and downlink frames are divided into $N^{u}$ time slots indexed by $n^{u} =\{1,\dots,N^{u}\}$ and $N^{d}$ time slots indexed by $n^{d} =\{1,\dots,N^{d}\}$, respectively. Moreover, each time slot contains one orthogonal
frequency division multiplexing (OFDM) symbol. Each user has one computation task ($B_{k}$, $D_{k}$) that needs to be processed, where $B_{k}$ is the task size in bits and $D_{k}$ is the time required for computation in time slots. Moreover, we assume a binary offloading scheme, where a task is executed as a whole either locally at the URLLC user or remotely at the MEC server. For task offloading, the user sends the task in the uplink and the edge server computes the task and sends the results back to the user in the downlink. There is an offset of $\tau$ time slots between downlink and uplink transmission. Thus, uplink and downlink transmission overlap in $\bar{O}=N^{u}-\tau$ time slots. The value of $\tau$ is a design parameter. On the one hand, if $\tau$ is chosen too small, the users' tasks may have not yet been computed when the downlink frame ends and hence the downlink resource is wasted. On the other hand, if $\tau$ is chosen too large, the computed bits at the BS have to wait before being transmitted to the users, which increases the end-to-end delay, see Fig. \ref{model}. The maximum transmit power of the BS is $P_{\text{max}}$, while the maximum transmit power of each user in the uplink is $P_{k,\text{max}}$. 

In order to facilitate the presentation, in the following, we use superscript
$j \in \{u,d\}$ to denote uplink $u$ and downlink $d$.
\begin{remark}
	We note that the time and power consumed for channel estimation and resource allocation are constant and do not affect the proposed resource allocation algorithm. For simplicity of illustration, they are neglected in this paper. Furthermore, perfect channel state information (CSI) is assumed to be available at the BS for resource allocation design to obtain a performance upper bound for OFDMA-URLLC MEC systems.	\vspace{-0.45cm}
\end{remark}
\begin{figure}[t]
	\centering
	\scalebox{0.34}{
		\pstool{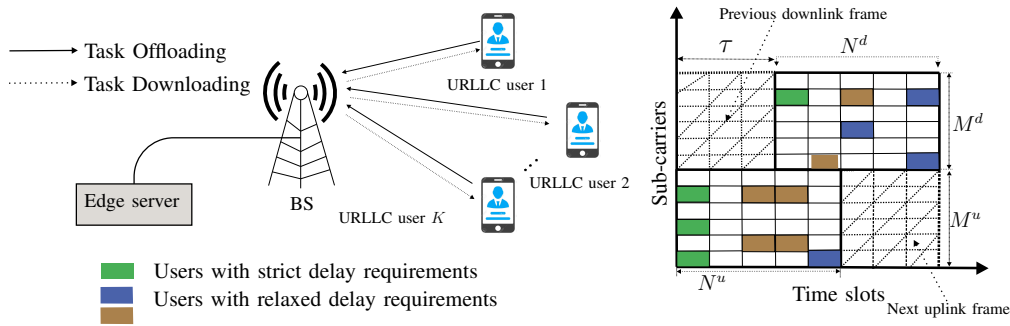}{
			\psfrag{e}[c][c][2]{Task Offloading}
			\psfrag{f}[c][c][2]{Sub-carriers}
			\psfrag{g}[c][c][2]{Task Downloading}
			\psfrag{a}[c][c][1.5]{URLLC user 1}	
			\psfrag{o}[c][c][2.5]{$\tau$}
			\psfrag{b}[c][c][1.5]{URLLC user 2}		
			\psfrag{c}[c][c][1.5]{\;\;URLLC user $K$}
			\psfrag{n}[c][c][1.8]{BS}
			\psfrag{s}[c][c][1.8]{Edge server}
			\psfrag{r}[c][c][2]{$\;\;$Users with relaxed delay requirements}
			\psfrag{h}[c][c][2]{Users with strict delay requirements}
			\psfrag{t}[c][c][2]{Time slots}
		    \psfrag{tu}[c][c][2]{$N^{u}$}
		    \psfrag{td}[c][c][2]{$N^{d}$}
		    \psfrag{wd}[c][c][2]{$M^{d}$}
	        \psfrag{wu}[c][c][2]{$M^{u}$}
	        \psfrag{nd}[c][c][1.5]{Previous downlink frame}
	        \psfrag{nu}[c][c][1.5]{Next uplink frame}
	}}
	\caption{ Multi-user MEC system comprising a single BS with an edge server and $K$ URLLC users.}
	\label{model}
	\vspace{-0.75cm}
\end{figure}
\subsection{Uplink and Downlink Channel Models}
In the following, we introduce the uplink and downlink channel models for the considered OFDMA-URLLC MEC system. We assume that the channel gains of all sub-carriers are constant for all users during uplink and downlink transmission. 
In the uplink, the signal received at the BS from user $k$ on sub-carrier $m^{u}$ in time slot $n^{u}$ is given as follows:	\vspace{-0.5cm}
\begin{equation}\label{su2a}\hspace{-0.65cm}
y^{u}_{k}[m^{u},n^{u}]=h_{k}^{u}[m^{u}]x^{u}_{k}[m^{u},n^{u}]+z^{u}_{BS}[m^{u},n^{u}],
\end{equation}
where $x^{u}_{k}[m^{u},n^{u}]$ denotes the symbol transmitted by user $k$ on sub-carrier $m^{u}$ in time slot $n^{u}$ to the BS. Moreover, $z^{u}_{BS}[m^{u},n^{u}]\sim \mathcal{CN}(0,\sigma^{2})$ denotes the noise on sub-carrier $m^{u}$ in time slot $n^{u}$ at the BS\footnote{Without loss of generality, we assume that the noise processes at all receivers have identical variances.}, and $h^{u}_{k}[m^{u}]$ represents the complex channel coefficient between user $k$ and the BS on sub-carrier $m^{u}$. For future reference, we define the signal-to-noise ratio (SNR) of user $k$'s signal at the input of the BS's receiver on sub-carrier $m^{u}$ in time slot $n^{u}$ as follows:	\vspace{-0.5cm}
\begin{equation}
\gamma^{u}_{k}[m^{u},n^{u}]=g^{u}_{k}[m^{u}]p^{u}_{k}[m^{u},n^{u}],
\end{equation}
where $p^{u}_{k}[m^{u},n^{u}]=\mathcal{E}\{|x^{u}_{k}[m^{u},n^{u}]|^{2}\}$ is the uplink transmit power of user $k$ on sub-carrier $m^{u}$ in time slot $n^{u}$, and $g^{u}_{k}[m^{u}]=\frac{|h^{u}_{k}[m^{u}]|^{2}}{\sigma^{2}}$. A similar channel model is assumed for downlink transmission and the corresponding SNR at user $k$ on sub-carrier $m^{d}$ in time slot $n^{d}$ is denoted by $\gamma^{d}_{k}[m^{d},n^{d}]$. 
\subsection{Achievable Rate for FBT}
Shannon's capacity theorem, on which most conventional resource allocation designs are based, applies to the asymptotic case where the packet length approaches infinity and the decoding error probability goes to zero \cite{shannon}. Thus, it cannot be used for resource allocation design for URLLC systems, as URLLC systems have to employ short packets to achieve low latency, which makes decoding errors unavoidable. For the performance evaluation of FBT, the so-called normal approximation for short packet transmission was developed in \cite{thesis}. For parallel complex additive white Gaussian noise (AWGN) channels, the maximum number of bits $\Psi$ conveyed in a packet comprising $L_{p}$ symbols can be approximated as follows\cite[Eq. (4.277)]{thesis},\cite[Fig. 1]{Erseghe1}:\vspace{-0.5cm}
\begin{align}\label{normalapproximation}
\Psi=\sum_{l=1}^{L_{p}}\log_{2}(1+\gamma[l])-aQ^{-1}(\epsilon)\sqrt{\sum_{l=1}^{L_{p}}{\nu}[l]},
\end{align}
where $\epsilon$ is the decoding packet error probability, and $Q^{-1}(\cdot)$ is the inverse of the Gaussian Q-function with $
Q(x)=\frac{1}{\sqrt{2\pi}}\int_{x}^{\infty}\text{exp}{\left(-\frac{t^{2}}{2}\right)}\text{d}t$.
$\nu[l]=(1-{(1+\gamma[l])^{-2}})$
and $\gamma[l]$ are the channel dispersion \cite{thesis} and the SNR of the $l$-th symbol, respectively, and  $a=\log_{2}(\text{e})$. \color{black}

In this paper, we base the joint uplink-downlink resource allocation algorithm design for OFDMA-URLLC MEC systems on (\ref{normalapproximation}). By allocating several resource elements from the available resources to a given user, the number of offloaded and downloaded bits of the user can be adjusted.
\section{Problem Formulation}
In this section, we explain the offloading and downloading process and introduce the QoS requirements of the OFDMA-URLLC MEC users. Moreover, we formulate the proposed resource allocation algorithm design as an optimization problem.	\vspace{-0.25cm}
\subsection{Computing Modes}
In this section, we explain the different computing modes of the users. First, we explain the local computing at the users. Then, we explain the steps required for offloading to the edge server.
 
\textit{1) Local Computing Mode:} According to \cite[Eq. (1)]{CPUscheduling,Energyfre}, the power consumption of the central processing unit (CPU) comprises the dynamic power, short circuit power, and leakage power where the
dynamic power is much larger than the other two. As a result, similar to \cite{Energyfre}, we only consider the dynamic power for local execution. According to \cite{CPUscheduling,Energyfre,miettinen2010energy}, the total energy required for computing a task of length $B_{k}$ bits at user $k$ is given by:\vspace{-0.5cm}
\begin{IEEEeqnarray}{lll}\label{entask}
	E_{k}=\kappa c_{k} B_{k} f_{k}^{2},
\end{IEEEeqnarray} 	
where $f_{k}$ denotes the CPU frequency of the $k$-th user, $\kappa$ is the effective switched capacitance which depends on the chip architecture and is assumed to be identical for all users, $c_{k}$ is the number of cycles required for processing of one bit which depends on the type of  application and the CPU architecture \cite{miettinen2010energy}. A user can reduce its total energy consumption by reducing the  CPU frequency. However, the task computing latency also depends on the frequency and is given as follows:\vspace{-0.5cm}
\begin{IEEEeqnarray}{lll}\label{latency}
	t_{k}=\frac{c_{k}B_{k}}{f_{k}}.
\end{IEEEeqnarray} 
Combining (\ref{entask}) and (\ref{latency}), the local power consumption at user $k$ is given as follows:\vspace{-0.5cm}
\begin{IEEEeqnarray}{lll}\label{localpower}
	P^{l}_{k}=\kappa f_{k}^{3}.
\end{IEEEeqnarray}
A local user can adjust its CPU frequency to minimize its local power consumption subject to a required task computing latency. Alternatively, considering the limited capability of its CPU, a user may prefer to offload its task to the edge server instead. This process is explained in the following. 

\textit{2) Offloading and Downloading:} The edge computing process is performed as follows. First, the user offloads its data to the edge server in the uplink. Subsequently, the edge server processes this data and sends the results back in the downlink transmission to the user. Thus, uplink and downlink transmission should satisfy the following constraints: \vspace{-0.25cm}
\begin{align}&\label{offu1}\hspace{-0.25cm}
\mathrm{C1}:\Psi^{u}_{k}(\mathbf{s}_{k}^{u},\mathbf{p}_{k}^{u}) \geq (1-\alpha_{k}) B_{k},\forall k,\
\mathrm{C2}:\Psi^{d}_{k}(\mathbf{s}_{k}^{d},\mathbf{p}_{k}^{d})\geq (1-\alpha_{k})\Gamma_{k}B_{k},\forall k,
\end{align} 
where\vspace{-0.75cm}
\begin{equation}\label{offu2}
\Psi^{j}_{k}(\mathbf{s}_{k}^{j},\mathbf{p}_{k}^{j})	=F^{j}_{k}(\mathbf{s}_{k}^{j},\mathbf{p}_{k}^{j})-V^{j}_{k}(\mathbf{s}_{k}^{j},\mathbf{p}_{k}^{j}), \forall j,
\end{equation}
and\vspace{-0.75cm}
\begin{align}\label{b2}\hspace{-0.5cm}
F^{j}_{k}(\mathbf{s}_{k}^{j},\mathbf{p}_{k}^{j})=\sum_{m^{j}=1}^{M^{j}}\sum_{n^{j}=1}^{N^{j }}s^{j}_{k}[m^{j},n^{j}]\log_{2}(1+\gamma^{j}_{k}[m^{j},n^{j}]), \forall j,
\end{align} \vspace{-1cm}
\begin{align}\label{b3}\hspace{-0.4cm}
V^{j}_{k}(\mathbf{s}^{j}_{k},\mathbf{p}_{k}^{j})=aQ^{-1}(\epsilon^{j}_{k})\sqrt{{\sum_{m^{j}=1}^{M^{j}}\sum_{n^{j}=1}^{N^{j}}s^{j}_{k}[m^{j},n^{j}]\nu^{j}_{k}[m^{j},n^{j}]}}, \forall j.
\end{align} 
Here, $s^{j}_{k}[m^{j},n^{j}]=\{0,1\}, \forall m^{j},n^{j}, k, \forall j,$ are the sub-carrier assignment indicators. If sub-carrier $m^{j}$ is assigned
to user $k$ in time slot $n^{j}$, we have $s^{j}_{k}[m^{j},n^{j}]=1$, otherwise $s^{j}_{k}[m^{j},n^{j}]=0$. Furthermore, we assume that each sub-carrier is allocated to at most one user to avoid multiple access interference.  $\mathbf{s}_{k}^{j}$ and $\mathbf{p}_{k}^{j}$ are the collections of optimization variables $s^{j}_{k}[m^{j},n^{j}], \forall m^{j},n^{j}$, and $p^{j}_{k}[m^{j},n^{j}], \forall m^{j},n^{j}, \forall j$, respectively, and
$
\nu^{j}_{k}[m^{j},n^{j}]=(1-(1+\gamma^{j}_{k}[m^{j},n^{j}])^{-2})$. Constraints $\mathrm{C1}$ and $\mathrm{C2}$ guarantee the transmission of $(1-\alpha_{k})B_{k}$ bits in the uplink and $\Gamma_{k}(1-\alpha_{k})B_{k}$ bits in the downlink for user $k$, respectively, where parameter $\Gamma_{k},\forall k,$ specifies the ratio of the size of the computing result and the size of the offloaded task. The value of $\Gamma_{k}$ depends on the application type, e.g., $\Gamma_{k}> 1$ for augmented reality applications\cite{jointsubchannel}. Moreover, $\alpha_{k}=\{0,1\}$  is the binary mode selection variable, where $\alpha_{k}=1$ for local computing and $\alpha_{k}=0$ for edge computing offloading. 
\subsection{Causality and Delay}
In the following, we explain the causality and delay constraints in the considered OFDMA-URLLC MEC system.

\textit{1) Causality:} Downlink transmission cannot start for a given user before all data of this user has been received at the BS via the uplink. Furthermore, according to Fig.~\ref{model}, uplink and downlink transmission overlap in time slot $n^{u}=\tau+o$ or equivalently $n_{d}=o, \forall o=\{1,\dots,\bar{O}\}$. For the downlink, we need to ensure that for each user $k$, if overlapping time slot $n^{d}=\tau+ o$ is allocated to the uplink, no overlapping time slot with $n^{d}\leq o$ is allocated to the downlink. Exploiting the binary nature of variables $s^u_k[m^u,n^u]$ and $s^d_k[m^d,n^d]$, this condition can be imposed by the following set of linear inequality constraints:	\vspace{-0.25cm} 
\begin{align}&\label{Co}\hspace{-.5cm}
\mathrm{C3}:s_{k}^{u}[m^{u},\tau+o]+s^{d}_{k}[m^{d},o]\leq 1 , \forall k, \forall m^{u}, \forall m^{d}, \forall o=\{1,\dots,\bar{O}\}. \vspace{-0.25cm}
\end{align}	
 As can be seen from (\ref{Co}), if user $k$ uses sub-carrier $m^{u}$ in time slot $n^{u}=\tau +o$, then the downlink resources at and before time slot $n^{d}=o$ will be forced to be zero, i.e., no data is sent to user $k$.

\textit{2) Delay:} The delay of a computing task is limited by requiring the downlink transmission to be finished before $D_{k}-\tau$ time slots as follows\footnote{In this paper, we neglect the computing time and power consumption at the edge server, and we only focus on uplink and downlink transmission. This model is valid when the edge server has sufficient processing and computing resources to carry out the small tasks of the URLLC users with negligible delay.}:\vspace{-0.25cm}
\begin{align}\hspace{-3cm}
\mathrm{C4}:s^{d}_{k}[m^{d},n^{d}] = 0, \forall n^{d} \geq D_{k}-\tau.\vspace{-0.25cm}
\end{align} 
The total latency of a computing task is determined by $D_{k}$ and $\tau$. Note that the values of $D_{k}$ and $\tau$ are assumed to be known for resource allocation.
\subsection{Total System Power Consumption}
The total system power consumption includes the power consumption of the users and the BS. The power consumption of user $k$ is given as follows\cite{powermodel,Zhoumec,user_centric}:\vspace{-0.25cm}
\begin{align}&\label{pk}
\overline{P}_{k}=\kappa f_{k}^{3}+\delta_{k} \sum_{m^{u}=1}^{M^{u}}\sum_{n^{u}=1}^{N^{u}} s^{u}_{k}[m^{u},n^{u}]p^{u}_{k}[m^{u},n^{u}]+(1-\alpha_{k})P_{\text{cir}},
\end{align} 
where the first term in (\ref{pk}) accounts for the local computation power consumption in case of local computing, the second term accounts for the power consumed for offloading transmission, and the third term accounts for the constant circuit power consumption during offloading. To model the  inefficiency of the power amplifiers of the users,
we introduce the multiplicative constant, $\delta_{k} \geq 1$, for the power radiated by the transmitter in (\ref{pk}) which takes into account the joint effect of the drain efficiency and backoff of the power amplifier \cite{poweramplifier}. Note that, as can be seen from $\mathrm{C1}$ and $\mathrm{C2}$, when $\alpha_{k}=1$, the required offloaded and downloaded data is zero, and hence, in this case, since we minimize the total power consumption, the power allocated for uplink transmission, $p^{u}_{k}[m^{u},n^{u}]$, will be zero $\forall m^{u}, \forall n^{u}$. On the other hand, for offloading, i.e., $\alpha_{k}=0$, the optimization problem formulated in the next subsection will ensure that the power consumption for local computing will be zero. Hence, there is no need to explicitly multiply the first and second term in (\ref{pk}) by $\alpha_{k}$ and $(1-\alpha_{k})$ to ensure that the terms are zero for offloading and local computing, respectively. Furthermore, due the significant computational resources of the BS, we neglect the corresponding computation power consumption. Moreover, since in practice the BS does not only serve the MEC users considered for resource allocation but also non-MEC users, the BS circuit power consumption is also not considered for optimization. Thus, the relevant weighted system power consumption is modelled as follows:
\begin{align}&\label{totalpower}
\Phi=\sum_{k=1}^{K}w_{k}\overline{P}_{k}+\delta_{\textrm{BS}}\sum_{k=1}^{K}\sum_{m^{d}=1}^{M^{d}}\sum_{n^{d}=1}^{N^{d}}s^{d}_{k}[m^{d},n^{d}]p^{d}_{k}[m^{d},n^{d}],
\end{align} 
where the second term in (\ref{totalpower}) represents the power consumption of the BS for downlink transmission and $\delta_{\textrm{BS}}\geq 1$ accounts for the inefficiency of the BS power amplifier. Moreover, $w_{k}\geq 1, \forall k,$ are weights that allow the prioritization of the users' power consumption compared to the BS's power consumption.     
\subsection{Optimization Problem Formulation}
In the following, we formulate the resource allocation problem with the goal to minimize the total weighted network power consumption, while satisfying the latency requirements of the users' computing tasks. In particular, we optimize the uplink and downlink transmit powers, the uplink and downlink sub-carrier assignment, the CPU frequency of the local CPUs, and the mode selection of each user. To this end, the optimization problem is formulated as follows: \vspace{-0.5cm}
\begin{align}\label{Op1}
&\hspace{-0.2cm}\underset {\boldsymbol{f}, \mathbf{s}^{u}, \mathbf{p}^{u}, \mathbf{s}^{d}, \mathbf{p}^{d}, \boldsymbol{\alpha}}{\text{minimize}}\Phi\\&\text{s.t.}\; \mathrm{C1-C4}, \;\mathrm{C5}:\sum_{k=1}^{K}s^{u}_{k}[m^{u},n^{u}] \leq 1, \forall m^{u},n^{u},\;\;
\mathrm{C6}:s^{u}_{k}[m^{u},n^{u}] \in \{0,1\}, \forall k,m^{u},n^{u},\nonumber\\&  \quad\;\;
\mathrm{C7}:\sum_{k=1}^{K}s^{d}_{k}[m^{d},n^{d}] \leq 1, \forall m^{d},n^{d}, \;\; \mathrm{C8}: s^{d}_{k}[m^{d},n^{d}] \in \{0,1\}, \forall k,m^{d},n^{d}, \nonumber\\&  \nonumber\quad\;\;
\mathrm{C9}:\sum_{m^{u}=1}^{M^{u}}\sum_{n^{u}=1}^{N^{u}}s^{u}_{k}[m^{u},n^{u}]p^{u}_{k}[m^{u},n^{u}]\leq P_{k,\text{max}}, \forall k, \;  \mathrm{C10}: p^{u}_{k}[m^{u},n^{u}] \geq 0, \forall k, m^{u}, n^{u},\; \\&  \quad\;\;
\mathrm{C11}:\sum_{k=1}^{K}\sum_{m^{d}=1}^{M^{d}}\sum_{n^{d}=1}^{N^{d}}s^{d}_{k}[m^{d},n^{d}]p^{d}_{k}[m^{d},n^{d}] \leq P_{\text{max}},\;\; \mathrm{C12}: p^{d}_{k}[m^{d},n^{d}] \geq 0, \forall k, m^{d}, n^{d},\nonumber
\\&  \quad\;\; \mathrm{C13}: c_{k}\alpha_{k}B_{k} \leq T_{s} f_{k} D_{k}, \forall k,  \; \mathrm{C14}: \alpha_{k} \in \{0,1\}, \forall k,  \;\;
\mathrm{C15}:0 \leq f_{k} \leq f_{\text{max}}, \forall k. \nonumber
\end{align}
Here, $\boldsymbol{f}$, $\mathbf{s}^{u}$, $\mathbf{p}^{u}$, $\mathbf{s}^{d}$, $\mathbf{p}^{d}$, and $\boldsymbol{\alpha}$ are the collections of optimization variables $f_{k}, \forall k$, $\mathbf{s}_{k}^{u}, \forall k$, $\mathbf{p}_{k}^{u}, \forall k$, $\mathbf{s}_{k}^{d}, \forall k$, $\mathbf{p}_{k}^{d}, \forall k$, and $\alpha_{k}, \forall k$, respectively. 

In (\ref{Op1}), constraints $\mathrm{C1}$ and $\mathrm{C2}$ guarantee the transmission of the required number of bits from user $k$ to the BS in the uplink and from the BS to user $k$ in the downlink, respectively, if the user offloads the task, i.e., $\alpha_{k}=0$. Constraint $\mathrm{C3}$ is the uplink-downlink causality constraint and constraint $\mathrm{C4}$ ensures that user $k$ is served such that its task meets the associated delay requirements. Constraints $\mathrm{C5}$ and $\mathrm{C6}$ for the uplink and constraints $\mathrm{C7}$ and $\mathrm{C8}$ for the downlink are imposed to ensure that each sub-carrier in a given time slot is allocated to at most one user. Constraints $\mathrm{C9}$ and $\mathrm{C11}$ are the total transmit power constraints of user $k$ and the BS, respectively. Constraints $\mathrm{C10}$ and $\mathrm{C12}$ are the non-negative transmit power constraints. Constraint $\mathrm{C13}$ ensures that the maximum allowed delay for local computing is not exceed when $\alpha_{k}=1$. 
 Constraint $\mathrm{C14}$ is the mode selection constraint. Finally, constraint $\mathrm{C15}$ limits the CPU frequency of the local CPUs to $f_{\text{max}}$.
% Moreover, as can be seen in $\mathrm{C1}$ and $\mathrm{C2}$, when $\alpha_{k}=1$, the required offloaded and downloaded data is zero, and in this case, since we minimize the total power consumption, the power allocated for transmission and reception, and the circuit power consumption at the users and the power consumption of the BS will be zero. 
\begin{remark}
	Resource allocation algorithm design for conventional MEC
	systems is typically based on Shannon's capacity formula, i.e., $V^{u}_{k}(\mathbf{s}_{k}^{u},\mathbf{p}_{k}^{u})$ and $V^{d}_{k}(\mathbf{s}_{k}^{d},\mathbf{p}_{k}^{d})$ in $\mathrm{C1}$ and $\mathrm{C2}$ are absent. The presence of $V^{u}_{k}(\mathbf{s}_{k}^{u},\mathbf{p}_{k}^{u})$ and $V^{d}_{k}(\mathbf{s}_{k}^{d},\mathbf{p}_{k}^{d})$ makes optimization problem (\ref{Op1}) significantly more difficult to solve but is essential for capturing the characteristics of OFDMA-URLLC MEC systems. 
\end{remark}
Problem (\ref{Op1}) is a mixed integer non-convex optimization problem. Such problems are in general NP hard and are known to be difficult to solve. However, in the next section, we propose an optimal scheme based on a branch-and-bound approach using monotonic optimization which finds the optimal solution of the considered problem. Moreover, in Section V, we propose two efficient suboptimal schemes that find close-to-optimal solutions
 and entail low computational complexity. 
\section{Proposed Global Optimal Solution}
In this section, we propose a branch-and-bound algorithm to solve problem (\ref{Op1}) optimally. Different from the general branch-and-bound algorithms proposed for non-convex problems, e.g., \cite{branch}, the proposed branch-and-bound algorithm exploits the monotonicity of the problem to reduce the search space for faster convergence\cite{biofast}. The purpose of finding a global optimal solution to (\ref{Op1}) is twofold: (1) determining a performance upper bound for OFDMA-URLLC MEC systems, and (2) having a benchmark for the efficient suboptimal solutions presented in Section V. We first introduce some mathematical background on monotonic optimization theory. Then, we transform optimization problem (\ref{Op1}) into the canonical form of discrete monotonic optimization. Finally, we present the optimal algorithm based on a new branch-and-bound algorithm which aims to minimize an upper bound on the objective function of (\ref{Op1}) until convergence to the optimal solution.  
\subsection{Mathematical Preliminaries for Monotonic Optimization}
In this subsection, we introduce some mathematical preliminaries for monotonic optimization\cite{yan,Tuy,Cheonmono,optimatravo}. 
\begin{defn}[{Increasing function}]
	A function $\psi :\mathbb{R}_{+}^{N\times 1} \rightarrow \mathbb{R} $ is increasing if $\psi(\mathbf{x}) \leq \psi(\mathbf{y})$ when $ 0 \leq \mathbf{x} \leq \mathbf{y}.$
\end{defn}
\begin{defn}[{Box}]
	Given any vector $\mathbf{x} \in \mathbb{R}_{+}^{N\times 1},$ the hyperrectangle $[\mathbf{\underline{x}}, {\mathbf{\overline{x}}}]=\{\mathbf{x}| \mathbf{\underline{x}} \leq \mathbf{x} \leq \mathbf{\overline{x}}\}$ is referred to as a box with lower and upper corners  $\mathbf{\underline{x}}$ and $\mathbf{\overline{x}}$, respectively.
\end{defn}
\begin{defn}[{Normal set}]
	A set  $\mathcal{G}\subset \mathbb{R}_{+}^{N\times 1}$ is normal if given any element $\mathbf{x} \in \mathcal{G}$, the box $[\mathbf{0},\mathbf{x}] \subset \mathcal{G}$.
\end{defn}
\begin{defn}[{Co-normal set}]
	A set  $\mathcal{H}$ is co-normal if $\mathbf{x}\in \mathcal{H}$ and $\mathbf{x}^{'}>\mathbf{x}$ imply $\mathbf{x}^{'} \in \mathcal{H}$. 
\end{defn}
\begin{defn}
	An optimization problem belongs to the class of discrete monotonic optimization problems if it can be represented in the following form\cite{yan,Tuy}:
	\begin{align}&
	\text{P1}: \underset {\mathbf {x}}{ \mathop {\mathrm {minimize}}\nolimits } ~\Lambda (\mathbf {x}) \qquad \mbox {s.t.}~\mathbf {x}\in \mathcal {V}=\mathcal {G} \cap \mathcal {H},
	\end{align}
	where $\Lambda(\mathbf {x})$ is an increasing function on $\mathbb{R}_{+}^{N\times 1}$ in $\mathbf {x}$ and $\mathcal{V}$ is a normal non-empty closed set, which is the intersection of normal set $\mathcal {G}$ and co-normal set $\mathcal {H}$.
\end{defn}
The solution of monotonic optimization problem $\text{P1}$ lies on the boundary of the feasible set\cite{Tuy}. As shown in \cite{Tuymonotonic,Tuy,Zhangmonotonic,emiloptimal,yan,yanmisoc,optimatravo}, the branch-and-bound algorithm can be used to iteratively approximate the boundary of the feasible set of $\text{P1}$ to find the global optimum solution in a finite number of iterations. In the following, we transform optimization problem (\ref{Op1}) into a monotonic optimization problem. Then, we propose an optimal algorithm based on the branch-and-bound technique.
\subsection{Problem Transformation}
In this subsection, we transform problem (\ref{Op1}) into the canonical form of a monotonic optimization problem. First, we introduce the following constraints in optimization problem (\ref{Op1}):    \vspace{-0.5cm}
\begin{align}&\label{prod}
\mathrm{C16}:{p}_{k}^{u}[m^{u},n^{u}]=s_{k}^{u}[m^{u},n^{u}] {p}_{k}^{u}[m^{u},n^{u}], \forall k,m^{u},n^{u},\\&\label{prod2a}
\mathrm{C17}:{p}_{k}^{d}[m^{d},n^{d}]=s_{k}^{d}[m^{d},n^{d}] {p}_{k}^{d}[m^{d},n^{d}], \forall k,m^{d},n^{d}.
\end{align}
Based on (\ref{prod}) and (\ref{prod2a}) optimization problem (\ref{Op1}) is transformed into the following equivalent form: 
\begin{align}\label{Op2}
&\underset {\boldsymbol{f}, \mathbf{s}^{u}, \mathbf{p}^{u}, \mathbf{s}^{d}, \mathbf{p}^{d}, \boldsymbol{\alpha}}{\text{minimize}}\sum_{k=1}^{K}w_{k}\bigg(\kappa f_{k}^{3}+\delta_{k}\sum_{m^{u}=1}^{M^{u}}\sum_{n^{u}=1}^{N^{u}}{p}^{u}_{k}[m^{u},n^{u}]+(1-\alpha_{k})P_{\text{cir}}\bigg)+\delta_{\textrm{BS}}\sum_{k=1}^{K}\sum_{m^{d}=1}^{M^{d}}\sum_{n^{d}=1}^{N^{d}}{p}^{d}_{k}[m^{d},n^{d}]\\& \nonumber\text{s.t.} \; \mathrm{C1}:	F^{u}_{k}( {\mathbf{p}}_{k}^{u})-V^{u}_{k}( {\mathbf{p}}_{k}^{u})\geq (1-\alpha_{k})B_{k},\forall k,\nonumber\;  \;\mathrm{C2}: F^{d}_{k}( {\mathbf{p}}_{k}^{d})-V^{d}_{k}( {\mathbf{p}}_{k}^{d})\geq (1-\alpha_{k})\Gamma_{k}B_{k} ,\forall k, \;\; \mathrm{C3-C8,}\nonumber\\ &\quad\;\; \;
\mathrm{C9}:\sum_{m^{u}=1}^{M^{u}}\sum_{n^{u}=1}^{N^{u}}{p}^{u}_{k}[m^{u},n^{u}] \leq P_{k,\text{max}}, \forall k, \;\; \mathrm{C10},\nonumber  \mathrm{C11}: \sum_{k=1}^{K}\sum_{m^{d}=1}^{M^{d}}\sum_{n^{d}=1}^{N^{d}}{p}^{d}_{k}[m^{d},n^{d}] \leq P_{\text{max}},\;\;\mathrm{C12-C17}.\nonumber
\end{align}
where\vspace{-0.75cm}
\begin{align}\label{b23}\hspace{-1.75cm}
F^{j}_{k}(\mathbf{p}_{k}^{j})=\sum_{m^{j}=1}^{M^{j}}\sum_{n^{j}=1}^{N^{j }}\log_{2}(1+\gamma^{j}_{k}[m^{j},n^{j}]), \quad 
V^{j}_{k}(\mathbf{p}_{k}^{j})=aQ^{-1}(\epsilon^{j}_{k})\sqrt{{\sum_{m^{j}=1}^{M^{j}}\sum_{n^{j}=1}^{N^{j}}\nu^{j}_{k}[m^{j},n^{j}]}}.
\end{align} 

Although optimization problem (\ref{Op2}) is still non-convex, it is more tractable compared to equivalent problem (\ref{Op1}), and as is shown in the following, it can be transformed into a monotonic optimization problem. To this end, we first study the monotonicity of problem (\ref{Op2}) in the following two lemmas. 
\begin{lem}
	Constraints $\mathrm{C1}$ and $\mathrm{C2}$ are differences of two monotonic and concave functions.
\end{lem}
\begin{proof}
	The proof closely follows a similar proof in \cite{ghanem1}, and is omitted here due to space limitation.
\end{proof}
\begin{lem}[see \cite{Tuymonotonic}]\label{lem2}
	Assume we have the following inequality $g(x)-h(x)\leq 0$, where both $g(x)$ and $h(x)$ are increasing functions. Assuming $ 0 \leq x \leq b$, then, $g(x) \leq g(b)$. Thus, there exist positive $t$ such that $g(x)+t \leq g(b)$. Therefore, the inequality $g(x)-h(x) \leq 0$ can be split into two inequalities $g(x)+t \leq g(b)$, $h(x)+t \geq g(b)$, where $0 \leq t \leq g(b)$.	
\end{lem} 
Therefore, based on Lemma \ref{lem2}, by defining positive auxiliary optimization variables $0 \leq \zeta_{k}^{u} \leq V^{u}_{k}(P_{k,\text{max}}), \forall k$, and $0 \leq \zeta_{k}^{d} \leq V^{d}_{k}(P_{\text{max}}), \forall k,$ we transform  non-monotonic constraints $\mathrm{C1}$ and $\mathrm{C2}$ into the following equivalent monotonic constraints:\vspace{-0.25cm}
\begin{align}&\label{1a}
\mathrm{C1a}:F^{u}_{k}(\mathbf{{p}}_{k}^{u})+\zeta_{k}^{u} \geq V^{u}_{k}(P_{k,\text{max}})+ (1-\alpha_{k})B_{k},\forall k,\;\;
\mathrm{C1b}:V^{u}_{k}(\mathbf{{p}}_{k}^{u})+\zeta_{k}^{u} \leq V^{u}_{k}(P_{k,\text{max}}),\forall k,\\&\label{2a}
\mathrm{C2a}:F^{d}_{k}(\mathbf{{p}}_{k}^{d})+\zeta_{k}^{d} \geq V^{d}_{k}(P_{\text{max}})+ (1-\alpha_{k})\Gamma_{k}B_{k},\forall k,
\;\;\mathrm{C2b}:V^{d}_{k}(\mathbf{{p}}_{k}^{d})+\zeta_{k}^{d} \leq V^{d}_{k}(P_{\text{max}}),\forall k,
\end{align}
where $V^{u}_{k}(P_{k,\text{max}})$ is obtained by allocating all  power available in the uplink, i.e., $P_{k,\text{max}}$, to time slot $n^{j}$, sub-carrier $m^{j}$, and user $k$. $V^{d}_{k}(P_{\text{max}})$ is defined in a similar way. Now, optimization problem (\ref{Op2}) can be transformed into the following equivalent form:\vspace{-0.25cm}
\begin{align}\label{Op3b}
&\hspace{-0.35cm}\underset {\boldsymbol{f}, \mathbf{s}^{u}, \mathbf{p}^{u}, \mathbf{s}^{d}, \mathbf{p}^{d}, \boldsymbol{\alpha}, \boldsymbol{\zeta}}{\text{minimize}}\sum_{k=1}^{K}w_{k}\bigg(\kappa f_{k}^{3}+\delta_{k}\sum_{m^{u}=1}^{M^{u}}\sum_{n^{u}=1}^{N^{u}}{p}^{u}_{k}[m^{u},n^{u}]+(1-\alpha_{k})P_{\text{cir}}\bigg)+\delta_{\textrm{BS}}\sum_{k=1}^{K}\sum_{m^{d}=1}^{M^{d}}\sum_{n^{d}=1}^{N^{d}}{p}^{d}_{k}[m^{d},n^{d}]\\
&\text{s.t.} \; \mathrm{C1a, C1b, C2a, C2b, C3-C17},\nonumber
\end{align}
where $\boldsymbol{\zeta}$ is the collection of optimization variables $\zeta_{k}^{j}, \forall k,j$. In order to find an optimal solution for (\ref{Op3b}), we perform an exhaustive search over the binary variables in $\boldsymbol{\alpha}$. For a given $\alpha_{k}=\bar{{\alpha}}_{k}, \forall k,$ optimization problem (\ref{Op3b}) reduces to the following optimization problem:  \vspace{-0.25cm}
\begin{align}\label{Op3c}
&\underset {\boldsymbol{f}, \mathbf{s}^{u}, \mathbf{p}^{u}, \mathbf{s}^{d}, \mathbf{p}^{d}, \boldsymbol{\zeta}}{\text{minimize}}\sum_{k=1}^{K}w_{k}\bigg(\kappa f_{k}^{3}+\delta_{k}\sum_{m^{u}=1}^{M^{u}}\sum_{n^{u}=1}^{N^{u}}{p}^{u}_{k}[m^{u},n^{u}]+(1-\bar{\alpha}_{k})P_{\text{cir}}\bigg)+\delta_{\textrm{BS}}\sum_{k=1}^{K}\sum_{m^{d}=1}^{M^{d}}\sum_{n^{d}=1}^{N^{d}}{p}^{d}_{k}[m^{d},n^{d}]\\
&\text{s.t.} \mathrm{C1a}:F^{u}_{k}(\mathbf{{p}}_{k}^{u})+\zeta_{k}^{u} \geq V^{u}_{k}(P_{k,\text{max}})+ (1-\bar{\alpha}_{k})B_{k},\forall k, \;\nonumber
\mathrm{C1b},\\&\quad \; \mathrm{C2a}:F^{d}_{k}(\mathbf{{p}}_{k}^{d})+\zeta_{k}^{d} \geq V^{d}_{k}(P_{\text{max}})+ (1-\bar{\alpha}_{k})\Gamma_{k}B_{k},\forall k,
\mathrm{C2b}, \mathrm{C3-C13},\nonumber
\\&\quad \; \mathrm{C15}: c_{k}\bar{\alpha}_{k}B_{k} \leq T_{s} f_{k} D_{k}, \forall k, \mathrm{C16, C17}.\nonumber
\end{align}
The optimal solution of problem (\ref{Op3b}) can be obtained by solving problem (\ref{Op3c}) for all  $2^{K}$ possible values of ${\boldsymbol{\alpha}}$ . Then, we select that ${\boldsymbol{\alpha}}=\bar{\boldsymbol{\alpha}}$ which minimizes the objective function of (\ref{Op3c}). Problem (\ref{Op3c}) is in the canonical form of a discrete monotonic optimization problem. Moreover, to facilitate the design of an optimal algorithm for solving (\ref{Op3c}), we rewrite (\ref{Op3c}) in the following form:
\begin{IEEEeqnarray}{lll}\label{Op34}& \underset {\boldsymbol{f}, \mathbf{s}^{u}, \mathbf{p}^{u},\mathbf{s}^{d}, \mathbf{p}^{d}, \boldsymbol{\zeta}}{ \mathop {\mathrm {minimize}}\nolimits }~	\bar{\Phi}
	\\&\;\; \mathrm {s.t.}~ \mathcal{V} \in \mathcal{G} \cap\mathcal{H},\nonumber
\end{IEEEeqnarray} 
where $\bar{\Phi}$ is the objective function in (\ref{Op3c}). Set $\mathcal{G}$ is defined by constraints $\mbox{C1b}, \mbox{C2b}$, and $\mbox{C3-C17}$, and co-normal set $\mathcal{H}$ is defined by constraints $\mbox {C1a}$ and $\mbox {C2a}$. The main difficulty in solving problem (\ref{Op34}) are the reverse convex constraints $\mbox{C1b}$, $\mbox{C2b}$, and the non-convex binary constraints $\mbox{C6}$ and $\mbox{C8}$. Moreover, for given $(\boldsymbol{f}, \mathbf{p}^{u}, \mathbf{p}^{d}, \boldsymbol{\zeta})$, problem (\ref{Op34}) can be solved optimally in the remaining variables as we will explain in the following. Therefore, an efficient algorithm to find the optimal solution of (\ref{Op34}) can be constructed by dividing optimization variables $\boldsymbol{f}$, $\mathbf{s}^{u}, \mathbf{p}^{u}, \mathbf{s}^{d}, \mathbf{p}^{d}$, and $\boldsymbol{\zeta}$ into two sets. The first set contains the convex variables $\boldsymbol{f}$ and $\boldsymbol{\zeta}$ and the non-convex variables $\mathbf{{p}}^{u}$ and $\mathbf{{p}}^{d}$ as the so-called outer variables, while the second set contains the binary variables $\mathbf{s}^{u}$ and $ \mathbf{s}^{d}$ as the so-called inner variables. Furthermore, once $\mathbf{{p}}^{u}$ and $\mathbf{{p}}^{d}$ have been determined, according to (\ref{prod}), (\ref{prod2a}), we can obtain the values of $\mathbf{{s}}^{u}$ and $\mathbf{{s}}^{d}$ by comparing the values of the entries of $\mathbf{p}^{u}$ and $\mathbf{p}^{d}$  with zero. If the value of $p_{k}[m^{j},n^{j}]$ is greater than 0, this means that the corresponding $s_{k}[m^{j},n^{j}]=1$, otherwise $s_{k}[m^{j},n^{j}]=0$. Moreover, for given $\boldsymbol{f}$, $\mathbf{{p}}^{u}$, $\mathbf{{p}}^{d}$, and $\boldsymbol{\zeta}$, problem (\ref{Op34}) turns into the following feasibility check problem:\vspace{-0.75cm}
\begin{IEEEeqnarray}{lll}\label{Op34b}& \underset {\mathbf{s}^{u},\mathbf{s}^{d}}{ \mathop {\mathrm {minimize}}\nolimits }~	1
	\\&\;\; \mathrm {s.t.}~ \mathcal{V} \in \mathcal{G} \cap \mathcal{H}.\nonumber
\end{IEEEeqnarray}
Since the values of $\mathbf{{s}}^{u}$ and $\mathbf{{s}}^{d}$ are known, we can simply check the constraint in (\ref{Op34b}).
\subsection{Design of Optimal Algorithm}
Optimization problem (\ref{Op34}) is a discrete monotonic optimization problem which can be optimally solved via the branch-and-bound algorithm as explained in the following\cite{8008832,biofast}. To facilitate the presentation of the optimal solution, we collect optimization variables $(\boldsymbol{f}, \mathbf{p}^{u}, \mathbf{p}^{d}, \boldsymbol{\zeta})$ in vector $\mathbf{u} \in \mathbb{R}^{L}$, where $L=K+KM^{u}N^{u}+KM^{d}N^{d}+2K$. The solution of (\ref{Op34}) lies on the boundary of the feasible set, due to the monotonicity of the objective function and the constraints. However, the boundary of the feasible set is unknown. Thus, we approach the boundary by enclosing the feasible set $\mathcal{V}=\mathcal{G} \cap \mathcal{H}$ by an initial box $\mathcal{B}^{(0)}=[\mathbf{\underline{u}}^{(0)}\;\mathbf{\overline{u}}^{(0)}]$, where $\mathbf{\underline{u}}^{(0)}$ and $\mathbf{\overline{u}}^{(0)}$ are lower and upper bounds, respectively, for the collection of variables $\mathbf{u}$. We ensure $\mathbf{\underline{u}}^{(0)}$ and $\mathbf{\overline{u}}^{(0)}$ to be contained in $\mathcal{G} \setminus \mathcal{H}$ and $\mathcal{H}$, respectively. If this condition is not satisfied, either the problem is infeasible (when $\mathbf{\underline{u}}^{(0)}$ is not in set $\mathcal{G}$) or $\mathbf{\underline{u}}^{(0)}$ is an optimal solution of the problem (when $\mathbf{\underline{u}}^{(0)}$ is in $\mathcal{V}$). Iteratively, we split certain hyperrectangles, i.e., boxes, on the optimization variables $\mathbf{u}$ and try to improve a lower bound and an upper bound on the optimal value of the objective function. To aid this process, a local lower bound $L_{\mathcal{B}}$ is stored for each box $\mathcal{B} \in \mathcal{L},$ where $\mathcal{L}$ is the set of all available boxes. Moreover, the current best value of the objective function obtained so far is denoted by $C_{BV}$. An algorithmic description of the proposed branch-and-bound scheme is presented in \textbf{Algorithm}~1. In the following, we explain the algorithm in more detail.

\textit{1) Selection and Branching:} In each iteration $i$ of the optimal algorithm, i.e., in Line 3 of \textbf{Algorithm} 1, we start by selecting the box $\mathcal{B}^{(i)}$ that has the lowest lower bound from the set of available boxes $\mathcal{L}$ as follows:  
\begin{equation}
\mathcal{B}^{(i)}=\mathop{\arg\min}_{{\cal B}\in{\cal L}}\bar{\Phi}(\mathbf{\underline{u}}).
\end{equation}
After selecting a box $\mathcal{B}^{(i)}=[\mathbf{\overline{u}}^{(i)}\; \mathbf{\underline{u}}^{(i)}]$, we bisect the longest edge of $\mathcal{B}^{(i)}$. We first calculate
\begin{equation}
\tilde{j}=\arg \max_{j=1\dots,L}\{[\underline{u}^{(i)}_{j}- \overline{u}^{(i)}_{j}]\},
\end{equation}
then, $\mathcal{B}^{(i)}$ is partitioned into two new boxes as follows\cite{Cheonmono}:   
\begin{equation}\label{bis}
\mathcal{B}^{(i)}_{1}=[\mathbf{\underline{u}}^{(i)}, \mathbf{\overline{u}}^{(i)}-\bigg(\frac{{\overline{u}}^{(i)}_{\tilde{j}}-{\underline{u}}^{(i)}_{\tilde{j}}}{2}\bigg)\mathbf{e}_{\tilde{j}}],\quad
\mathcal{B}^{(i)}_{2}=[\mathbf{\underline{u}}^{(i)}+\bigg(\frac{{\overline{u}}^{(i)}_{\tilde{j}}-{\underline{u}}^{(i)}_{\tilde{j}}}{2}\bigg)\mathbf{e}_{\tilde{j}}, \mathbf{\overline{u}}^{(i)}],
\end{equation}
where $\mathbf{e}_{\tilde{j}} \in \mathbb{R}^{L}$ is a vector whose $\tilde{j}$-th element is equal to one and the remaining elements are zero. The bisection rule in (\ref{bis}) guarantees that the branching process is exhaustive \cite{global,Cheonmono,biofast} and the algorithm converges to the optimal solution. 
%\begin{figure}[t]
%	\centering
%	\scalebox{0.9}{
%		\pstool{Figs/brb2.eps}{
%			\psfrag{x}[c][c][0.8]{1) Select}
%			\psfrag{y}[c][c][0.8]{2) Branch}
%			\psfrag{z}[c][c][0.8]{3) Feasability check}
%		    \psfrag{j}[c][c][0.8]{Set boundary}
%		    \psfrag{a}[c][c][0.8]{$\mathbf{\underline{u}}^{(i)}$}
%		    \psfrag{b}[c][c][0.8]{$\mathbf{\overline{u}}^{(i)}$}
%		    \psfrag{e}[c][c][0.8]{$\mathbf{\underline{u}}^{(i)}_{1}$}
%		    \psfrag{f}[c][c][0.8]{$\mathbf{\underline{u}}^{(i)}_{2}$}
%		    \psfrag{g}[c][c][0.8]{$\mathbf{\overline{u}}^{(i)}_{1}$}
%		    \psfrag{h}[c][c][0.8]{$\mathbf{\overline{u}}^{(i)}_{2}$}
%		    \psfrag{l2}[c][c][0.8]{$\mathbf{\underline{u}}^{(i)}_{1}$}
%		    \psfrag{l3}[c][c][0.6]{$L_{\mathcal{B},1}^{(i)}=\bar{\Phi}(\mathbf{\underline{u}}_{1}^{(i)}) $}
%		    \psfrag{l4}[c][c][0.8]{$L_{\mathcal{B},2}^{(i)}=\bar{\Phi}(\mathbf{\underline{u}}_{2}^{(i)}) $}
%		    \psfrag{l1}[c][c][0.8]{$\mathbf{\underline{u}}^{(i)}_{2}$}
%		    \psfrag{u2}[c][c][0.8]{$\mathbf{\overline{u}}^{(i)}_{1}$}
%		    \psfrag{u1}[c][c][0.8]{$\mathbf{\overline{u}}^{(i)}_{2}$}
%		    
%	}}
%	\caption{The main steps of the proposed branch-and-bound algorithm in a two dimensional space.}
%	\label{brb}
%	\vspace{-0.5cm}
%\end{figure}

\textit{2) Feasibility Check:} After the two new boxes $\mathcal{B}^{(i)}_{1}=[\mathbf{\underline{u}}_{1}^{(i)}\;\; \mathbf{\overline{u}}_{1}^{(i)}]$ and $\mathcal{B}^{(i)}_{2}=[\mathbf{\underline{u}}_{2}^{(i)}\;\; \mathbf{\overline{u}}_{2}^{(i)}]$ are generated, we check the lower and upper corners of each box and verify whether these boxes are feasible or not, see Lines 4-20. To do so, we first calculate local lower bounds $L_{\mathcal{B},b}^{(i)}=\bar{\Phi}(\mathbf{\underline{u}}_{b}^{(i)}), \forall b=\{1,2\}$ for $\mathcal{B}^{(i)}_{1}$ and $\mathcal{B}^{(i)}_{2}$, respectively, see Line 7. Subsequently, we compare the values of the local lower bounds $L_{\mathcal{B},b}^{(i)}, \forall b=\{1,2\}$ with the  best global value $C_{BV}$ obtained so far. If the local lower bound of one of the two new boxes is greater than $C_{BV}$, then this box can be removed. On the other hand, if the local lower bound is smaller than $C_{BV}$, we check the feasibility of the box and search for better feasible points. To do so, we first check the lower corners of each box by checking the feasibility of (\ref{Op34b}). If the lower corners are feasible, then, these lower corners will be added to the set of feasible solutions  $\mathcal{S}$ and we update the current best value $C_{BV}$. Otherwise, if this condition is not satisfied, we check if the box contains feasible solutions. The box is not feasible if $\mathbf{\underline{u}}^{(i)} \notin \mathcal{G}$ or $\mathbf{\overline{u}}^{(i)} \notin	\mathcal{H}$. In this case, we remove the infeasible box in the next step of the algorithm, i.e., in the pruning step. 

\begin{remark}
Although variables $\boldsymbol{\zeta}$ and $\boldsymbol{f}$ are convex variables, we branch over them. In fact, this facilitates the optimal algorithm design and reduces the total computation time needed for finding the optimal solution as it eliminates the use of convex software solvers which would contribute significantly to the overall computation time. 
\end{remark}

\textit{3) Bounding and Pruning:} The bounding and pruning steps are described in the following:

\textit{Bounding:} The problem is to find upper and lower bounds for $\bar{\Phi}(\mathbf{u})$ over the set $\mathcal{{G}} \cap \mathcal{H}$ for a given box $\mathcal{B}=[\mathbf{\underline{u}}\;\; \mathbf{\overline{u}}]$. Due to the monotonicity  of $\bar{\Phi}(\cdot)$ we can obtain the upper and lower bounds as $\bar{\Phi}(\mathbf{\overline{u}})$ and $\bar{\Phi}(\mathbf{\underline{u}})$, respectively.

\textit{Pruning:}  In the pruning step infeasible boxes are removed. These boxes have local lower bounds greater than the current best global value, i.e., 
$L_{\mathcal{B},b}^{(i)}>C_{BV}, \forall b$, and the original branched box in iteration $i$, i.e., $\mathcal{B}^{(i)}$. This step is performed to reduce memory consumption and to achieve faster convergence.
\subsection{Complexity Analysis}
For sufficiently large number of iterations $I_{\text{max}}$,  \textbf{Algorithm} 1 is guaranteed to find the optimal solution to optimization problem (\ref{Op1}). Its convergence can be proved using the same arguments as those in \cite{Tuy,8008832,Cheonmono}. However, the computational complexity of \textbf{Algorithm} 1 is exponential in the number of variables of the optimization problem. Thus, the complexity order of \textbf{Algorithm} 1 is $\mathcal{O}(2^{L})$. Due to its high complexity, the proposed optimal resource allocation algorithm cannot be used in real time applications, especially for URLLC systems. However, it provides a  valuable performance benchmark for  low-complexity suboptimal algorithms. Thus, in the next section, we focus on developing low-complexity resource allocation algorithms based on SCA to strike a balance between computational complexity and performance. 
\begin{algorithm}[t]
		\label{BRB1}
	\caption{Branch-and-bound algorithm}
	\begin{algorithmic}[1]
		\STATE {\textbf{Initialization:}} Ensure $\mathbf{\underline{u}}^{(0)} \in \mathcal{G} \setminus	 \mathcal{H}$ and $\mathbf{\overline{u}}^{(0)} \in \mathcal{H}$. Set $\mathcal{B}^{(0)}=[\mathbf{\underline{u}}^{(0)}, \mathbf{\overline{u}}^{(0)}]$,
%		 set error tolerances  $\delta$ and $\rho$, 
		 $\mathcal{L}=\{\mathcal{B}^{(0)}\}$, $L_{\mathcal{B}}(\mathcal{B}^{(0)})=\bar{\Phi}(\mathbf{\underline{u}}^{(0)})$, $C_{BV}(\mathcal{B}^{(0)})=\bar{\Phi}(\mathbf{\overline{u}}^{(0)})$, $\mathcal{S}$ denotes a set of feasible solutions, and maximum iteration number $I_{\text{max}}$.
		\STATE \textbf{for} $\;\;i=1: I_{\text{max}}$ 	
		\STATE {\textbf{Selection and branching:}} Select box $\mathcal{B}^{(i)}=[\mathbf{\underline{u}}^{(i)}\; \mathbf{\overline{u}}^{(i)}] \in \mathcal{L}$ such that $
		\mathcal{B}^{(i)}=\mathop{\arg\min}_{{\cal B}\in{\cal L}}\bar{\Phi}(\mathbf{\underline{u}})$ and branch it into two new boxes $\mathcal{B}^{(i)}_{1}$ and $\mathcal{B}^{(i)}_{2}$.\\
		\STATE \textbf{Feasibility check of the two new boxes:}
		\STATE  ${\textbf{for} \;\; b=1:2}$	
		\STATE $\;\;\;\;$
        suppose $\mathcal{B}_{b}^{(i)}=[\mathbf{\underline{u}}_{b}^{(i)}\; \mathbf{\overline{u}}_{b}^{(i)}]$
		\STATE $\;\;\;\;$ calculate local lower bound $L_{\mathcal{B},b}^{(i)}$ for $\mathcal{B}_{b}^{(i)}$ 
		\STATE $\;\;$ \textbf{if} ($L_{\mathcal{B},b}^{(i)}<C_{BV}$)
		\STATE $\;\;\;\;\;\;$  check the feasibility of lower corner $\mathbf{\underline{u}}_{b}^{(i)}$ by solving (\ref{Op34b}) 
		\STATE $\qquad$ \textbf{if} lower corner $\mathbf{\underline{u}}_{b}^{(i)}$ is feasible
		\STATE $\qquad\quad$  update $C_{BV}=L_{\mathcal{B},b}^{(i)}$  and store the feasible solution $\mathbf{\underline{u}}_{b}^{(i)}$, i.e., $\mathcal{S}\leftarrow \mathbf{\underline{u}}_{b}^{(i)}$, 
		\STATE $\;\;\;\;\;\;\;\;$ \textbf{else} 
		\STATE $\qquad\;\;$ \textbf{if} $\mathbf{\underline{u}}_{b}^{(i)} \in \mathcal{{G}}$ and $\mathbf{\overline{u}}_{b}^{(i)} \in \mathcal{H}$
		\STATE $\qquad\quad$ the box may be feasible, i.e., may contain feasible solutions 
		\STATE $\qquad\;\;$ \textbf{else}
	\STATE $\qquad\quad$	the box is not feasible and cannot contain any feasible solution 
		\STATE $\;\;\;\;\;\;\;\;\;\;$\textbf{end if}
		\STATE $\qquad$\textbf{end if}
		\STATE $\;\;$\textbf{end if}
		\STATE \textbf{end for}
		\STATE \textbf{Bounding and Pruning:} Update the set of boxes $\mathcal{L}$ for the next iteration of the algorithm  
		\STATE $\;\;$ \textbf{for} each $\mathcal{B} \in \mathcal{L}$ do
		\STATE $\quad$ \textbf{if} $L_{\mathcal{B},b}^{(i)}>C_{BV}$
		\STATE $\quad\;\;$  Remove $\mathcal{B}_{b}^{(i)}$
		\STATE $\quad$ \textbf{end if}
		\STATE $\;\;\;\;$ remove the branched box $(\mathcal{L} \leftarrow \mathcal{L}\setminus \mathcal{B}^{(i)})$ and remove infeasible boxes
%		\STATE $\;\;$ \textbf{end if}   
		\STATE $\;\;$ \textbf{end for}
		\STATE $i \leftarrow i + 1$
		\STATE \textbf{end for}
		\STATE \textbf{Output:} Optimal solution $\mathbf{u}^{*}$. 
	\end{algorithmic}
\end{algorithm}    
\section{SCA-Based Suboptimal Solutions}
In this section, we propose two low-complexity suboptimal algorithms based on SCA.
\subsection{Proposed SCA-Based Suboptimal Scheme 1}
In this sub-section, we propose a suboptimal algorithm that tackles the non-convexity of (\ref{Op1}) in three main steps. First, we use the Big-M formulation to linearize the product terms $s_{k}^{j}[m^{j},n^{j}] {p}_{k}^{j}[m^{j},n^{j}], \forall k,m^{j},n^{j}, \forall j$. Then, we employ difference of convex (DC) programming and SCA methods to find a locally optimal solution of optimization problem (\ref{Op1}).
      
 \textit{1) Big-M Formulation:} Let us first introduce the new optimization variables\footnote{For more details on the big M-formulation,
 	please refer to \cite[Section~2.3]{Leemixed}.} 
\begin{align}&\label{prod2}
\bar{p}_{k}^{j}[m^{j},n^{j}]=s_{k}^{j}[m^{j},n^{j}] {p}_{k}^{j}[m^{j},n^{j}], \forall k,m^{j},n^{j}, \forall j.
\end{align}
Now, we decompose the product term in (\ref{prod2}) using the Big-M formulation and impose the following additional constraints\cite{Leemixed}:
\begin{align}&
\mbox {C16}: \bar{p}^{u}_{k}[m^{u},n^{u}]\leq P_{k,\text{max}} s^{u}_{k}[m^{u},n^{u}], \forall k,m^{u},n^{u}, \;\;
\mbox {C17}: \bar{p}^{u}_{k}[m^{u},n^{u}]\leq p_{k}^{u}[m^{u},n^{u}], \forall k,m^{u},n^{u},\\ &  \mbox {C18}: \bar{p}_{k}^{u}[m^{u},n^{u}]\geq p_{k}^{u}[m^{u},n^{u}] -(1-s_{k}^{u}[m^{u},n^{u}])P_{k,\text{max}}, \forall k,m^{u},n^{u},\quad  \\ & \mbox {C19}: \bar{p}_{k}^{u}[m^{u},n^{u}]\geq 0,  \forall k,m^{u},n^{u},\;\; \mbox {C20}: \bar{p}^{d}_{k}[m^{d},n^{d}]\leq P_{\text{max}} s^{d}_{k}[m^{d},n^{d}],  \forall k,m^{d},n^{d}, \\ &  \mbox {C21}: \bar{p}^{d}_{k}[m^{d},n^{d}]\leq p_{k}^{d}[m^{d},n^{d}],  \forall k,m^{d},n^{d}, \;\; \mbox {C22}: \bar{p}_{k}^{d}[m^{d},n^{d}]\geq 0, \forall k,m^{d},n^{d}.\\ &  \mbox {C23}: \bar{p}_{k}^{d}[m^{d},n^{d}]\geq p_{k}^{d}[m^{d},n^{d}]-(1-s_{k}^{d}[m^{d},n^{d}])P_{\text{max}},    \forall k,m^{d},n^{d}.
\end{align}
In this manner, the non-convex product term $s_{k}^{j}[m^{j},n^{j}] {p}_{k}^{j}[m^{j},n^{j}], \forall k,m^{j},n^{j}, \forall j$ in (\ref{prod2}) is transformed into a set of convex linear inequalities. Note that constraints $\mbox {C16-C23}$ do not change the feasible set. Now, optimization problem (\ref{Op1}) is transformed into the following equivalent form: 
\begin{align}\label{Op2b}
	&\underset {\boldsymbol{f}, \mathbf{s}^{u}, \mathbf{p}^{u},\mathbf{s}^{d}, \mathbf{p}^{d},\mathbf{\bar{p}}^{u}, \mathbf{\bar{p}}^{d},  \boldsymbol{\alpha}}{\text{minimize}}\bar{\bar{\Phi}}\\
	&\text{s.t.} \; \overline{\mathrm{C1}}:	\bar{F}^{u}_{k}(\bar{\mathbf{p}}_{k}^{u})-\bar{V}^{u}_{k}(\bar{\mathbf{p}}_{k}^{u})\geq (1-\alpha_{k}) B_{k},\forall k,\;\overline{\mathrm{C2}}:		\bar{F}^{d}_{k}(\bar{\mathbf{p}}_{k}^{d})-\bar{V}^{d}_{k}(\bar{\mathbf{p}}_{k}^{d})\geq (1-\alpha_{k})\Gamma_{k}B_{k} ,\forall k, \mathrm{C3-C8},\nonumber\\&   \quad\;\;
  \mathrm{C9}: \sum_{m^{u}=1}^{M^{u}}\sum_{n^{u}=1}^{N^{u}}\bar{p}^{u}_{k}[m^{u},n^{u}] \leq P_{k,\text{max}}, \forall k, \; \mathrm{C10}, 	 \mathrm{C11}: \sum_{k=1}^{K} \sum_{m^{d}=1}^{M^{d}}\sum_{n^{d}=1}^{N^{d}}\bar{p}^{d}_{k}[m^{d},n^{d}] \leq P_{\text{max}}, \forall k, \;
	\mathrm{C12-C23}.\nonumber
\end{align}
where \vspace{-0.25cm}
\begin{align}&\label{obj1}
\bar{\bar{\Phi}}=\sum_{k=1}^{K}w_{k}\bigg(\kappa f_{k}^{3}+\delta_{k}\sum_{m^{u}=1}^{M^{u}}\sum_{n^{u}=1}^{N^{u}}\bar{p}^{u}_{k}[m^{u},n^{u}]+(1-\alpha_{k})P_{\text{cir}}\bigg)+\delta_{\textrm{BS}}\sum_{k=1}^{K}\sum_{m^{d}=1}^{M^{d}}\sum_{n^{d}=1}^{N^{d}}\bar{p}^{d}_{k}[m^{d},n^{d}],
\end{align} \vspace{-0.75cm}
\begin{align}&\label{c1}
\bar{F}^{j}_{k}(\mathbf{\bar{p}}_{k}^{j})=\sum_{m^{j}=1}^{M^{j}}\sum_{n^{j}=1}^{N^{j }}\log_{2}(1+\bar{\gamma}^{j}_{k}[m^{j},n^{j}]),\;\;
\bar{V}^{j}_{k}(\mathbf{\bar{p}}_{k}^{j})=aQ^{-1}(\epsilon^{j}_{k})\sqrt{\sum_{m^{j}=1}^{M^{j}}\sum_{n^{j}=1}^{N^{j}}\bar{\nu}^{j}_{k}[m^{j},n^{j}]},
\end{align} 
$\bar{\gamma}^{j}_{k}[m^{j},n^{j}]=g^{j}_{k}[m^{j}]\bar{p}^{j}_{k}[m^{j},n^{j}],$ and
$\bar{\nu}^{j}_{k}[m^{j},n^{j}]=(1-(1+\bar{\gamma}^{j}_{k}[m^{j},n^{j}])^{-2}).$ Moreover, $\mathbf{\bar{p}}_{k}^{j}, \forall j$ are the collection of optimization variables $\bar{p}_{k}^{}[m^{j},n^{j}], \forall m^{j}, n^{j},$ and $\mathbf{\bar{p}}^{j},$ are the collection of optimization variables $\mathbf{\bar{p}}_{k}^{j}, \forall k,$ where $j \in \{u,d\}$.

 \textit{2) DC Programming:} The two remaining difficulties for solving problem (\ref{Op2b}) are the binary variables in constraints $\mathrm{C6}$, $\mathrm{C8}$, and $\mathrm{C14}$ and the structure of the achievable rate for FBT in $\mathrm{C1}$ and $\mathrm{C2}$. To tackle these issues, we employ a DC programming approach\cite{ghanem1,yan,kwan1,Joinoptimization}. To this end, the integer constraints  in (\ref{Op2b}) are rewritten in the following DC function forms:\vspace{-0.5cm}
\begin{equation}\label{eq1}
\mathrm{C6a}: 0 \leq s^{u}_{k}[m^{u},n^{u}] \leq 1, \forall k,m^{u},n^{u}, 
\quad
\mathrm{C6b}:  E^{u}(\mathbf{s}^{{u}}) -H^{u}(\mathbf{s}^{{u}}) \leq 0,
\end{equation}
\begin{equation}\label{eq2}
\mathrm{C8a}: 0 \leq s^{d}_{k}[m^{d},n^{d}] \leq 1, \forall k,m^{d},n^{d}, \quad
\mathrm{C8b}:  E^{d}(\mathbf{s}^{{d}}) -H^{d}(\mathbf{s}^{{d}}) \leq 0,
\end{equation}
\begin{equation}\label{eq3}\hspace{0.3cm}
\mathrm{C14a}: 0 \leq \alpha_{k} \leq 1, \forall k, \qquad\qquad\quad\qquad
\mathrm{C14b}:  E^{\alpha}(\boldsymbol{\alpha}) -H^{\alpha}(\boldsymbol{\alpha}) \leq 0,
\end{equation}
where\vspace{-0.4cm}
\begin{align}&\label{esu}
E^{j}(\mathbf{s}^{{j}})=\sum_{k=1}^{K}\sum_{m^{j}=1}^{M^{j}}\sum_{n^{j}=1}^{N^{j}}s^{j}_{k}[m^{j},n^{j}], \forall j, \qquad H^{j}(\mathbf{s}^{{j}})=\sum_{k=1}^{K}\sum_{m^{j}=1}^{M^{j}}\sum_{n^{j}=1}^{N^{j}}(s^{j}_{k}[m^{j},n^{j}] )^{2}, \forall j.    
\end{align} 
and
\begin{equation}\label{eqb3}
E^{\alpha}(\boldsymbol{\alpha})=\sum_{k=1}^{K}\alpha_{k},\qquad\qquad\quad\qquad
H^{\alpha}(\boldsymbol{\alpha})= \sum_{k=1}^{K} \alpha_{k}^{2}.
\end{equation}
Now, constraints $\mathrm{C6}$, $\mathrm{C8}$, and $\mathrm{C14}$ are equivalently formulated in continuous form, cf. $\mathrm{C6a}$, $\mathrm{C8a}$, and $\mathrm{C14a}$. However, constraints $\mathrm{C6b}$, $\mathrm{C8b}$, and $\mathrm{C14b}$ are still non-convex, i.e., reverse convex constraints. In order to handle them, we introduce the following lemma.
\begin{lem}\label{lemmapena}
	For sufficiently large constant values $\eta_{1}$, $\eta_{2}$, and $\eta_{3}$, problem (\ref{Op2b}) is equivalent to the following problem:\vspace{-0.75cm}
	\begin{align}\label{Op3}& \underset {\boldsymbol{f},\mathbf{s}^{u}, \mathbf{p}^{u},\mathbf{s}^{d}, \mathbf{p}^{d},\mathbf{\bar{p}}^{u}, \mathbf{\bar{p}}^{d},  \boldsymbol{\alpha}}{ \mathop {\mathrm {minimize}}\nolimits }~\bar{\bar{\Phi}}+{\beta}(\mathbf{s}^{u},\mathbf{s}^{d},\boldsymbol{\alpha})
		\\&\;\; \mathrm {s.t.}~\overline{\mathrm{C1}},\overline{\mathrm{C2}}, \mathrm{C3-C5, C6a, C7, C8a, C9, C10-C15, C14a, C17-C23},\nonumber
	\end{align} 
where \vspace{-1cm} 
\begin{align}&\label{betabar}\hspace{-1cm} 
{\beta}(\mathbf{s}^{u},\mathbf{s}^{d},\boldsymbol{\alpha})=\eta_{1}(E^{u}(\mathbf{s}^{{u}})-H^{u}(\mathbf{s}^{{u}}))+\eta_{2}(E^{d}(\mathbf{s}^{{d}})-H^{d}(\mathbf{s}^{{d}}))+\eta_{3}(E^{\alpha}(\boldsymbol{\alpha})-H^{\alpha}(\boldsymbol{\alpha})).
\end{align}    
\end{lem}           
\begin{proof}
%The proof follows similar steps as corresponding proofs in \cite{gha4,ghanem1, yan,kwan1}. Interested readers are referred to [2] which is an extended version of this paper and includes the full proof.
Please refer to Appendix A.
\end{proof}
Constants $\eta_{1}$, $\eta_{2}$, and $\eta_{3}$ act as penalty factors 
to penalize the objective function for any $s_{k}^{j}[m^{j},n^{j}]$ that is not equal to 0 or 1.  
The remaining sources of non-convexity are the structure of the achievable rate for FBT and the non-convex objective function. In the following, we employ SCA to approximate problem (\ref{Op3}) by a convex problem. Subsequently, we propose an iterative algorithm to find a low-complexity solution.      

 \textit{3) SCA:}
In order to tackle the remaining non-convexity of (\ref{Op3}), we employ the Taylor series approximation to approximate the non-convex parts of the objective function and constraints $\overline{\mathrm{C1}}$ and $\overline{\mathrm{C2}}$. Since $H^{j}(\mathbf{s}^{{j}}),\forall j$, $-{V}^{j}_{k}(\mathbf{\bar{p}}^{j}_{k}),\forall j$, and $H^{\alpha}(\boldsymbol{\alpha})$ are differentiable convex functions, then for any feasible points $\mathbf{s}^{{j(i)}}, \mathbf{\bar{p}}^{j(i)}_{k}, \forall j$, and $\boldsymbol{\alpha}^{(i)}$, where the superscript $i$ denotes the SCA iteration index, the following inequalities hold:          
\begin{align}&\label{htayslorhu}
H^{j}(\mathbf{s}^{{j}}) \ge\bar{H}^{j}(\mathbf{s}^{{j}},\mathbf{s}^{{j(i)}})= H^{j}(\mathbf{s}^{{j(i)}}) +\nabla _{\mathbf {s}^{j}} H^{j}(\mathbf {s}^{j(i)})^{T}(\mathbf{s}^{{j}}-\mathbf{s}^{{j(i)}}), \forall j, \\&
\label{vua}
{V}^{j}_{k}(\mathbf{\bar{p}}^{j}_{k}) \leq \bar{V}^{j}_{k}(\mathbf{\bar{p}}^{j}_{k},\mathbf{\bar{p}}^{j(i)}_{k}) =  {V}^{j}_{k}(\mathbf{\bar{p}}^{j(i)}_{k})+ \nabla_{\mathbf{\bar{p}}^{j}_{k}}{V}_{k}(\mathbf{\bar{p}}^{j(i)}_{k})^{T}(\mathbf{\bar{p}}^{j}_{k}-\mathbf{\bar{p}}^{j(i)}_{k}), \forall j.
\end{align}
and \vspace{-0.5cm}
\begin{align}&\label{halph}
H^{\alpha}(\boldsymbol{\alpha}) \ge 	\bar{H}^{\alpha}(\boldsymbol{\alpha},\boldsymbol{\alpha}^{(i)})= H(\boldsymbol{\alpha}^{(i)})  +\nabla _{\boldsymbol{\alpha}} H(\boldsymbol{\alpha}^{(i)})^{T}(\boldsymbol{\alpha}-\boldsymbol{\alpha}^{(i)}).
\end{align}
The right hand sides of (\ref{htayslorhu}), (\ref{vua}), and (\ref{halph}) are affine functions representing the global underestimation of $H^{j}(\mathbf{s}^{{j}}), \forall j$,  ${V}^{j}_{k}(\mathbf{\bar{p}}^{j}_{k}),\forall j$, and $H^{\alpha}(\boldsymbol{\alpha})$, respectively, where $\nabla _{\mathbf {s}^{j}} H^{j}(\mathbf {s}^{j(i)})$ and $\nabla_{\mathbf{\bar{p}}^{j}_{k}}{V}_{k}(\mathbf{\bar{p}}^{j(i)}_{k})$ are the gradients of $H^{j}(\mathbf{s}^{{j}})$ and ${V}^{j}_{k}(\mathbf{\bar{p}}^{j}_{k})$, respectively.
%The right hand sides of (\ref{htayslorhu}), (\ref{vua}), and (\ref{halph}) are affine functions representing the global underestimation of $H^{j}(\mathbf{s}^{{j}}), \forall j$,  ${V}^{j}_{k}(\mathbf{\bar{p}}^{j}_{k}),\forall j$, and $H^{\alpha}(\boldsymbol{\alpha})$, respectively, where $\nabla _{\mathbf {s}^{j}} H^{j}(\mathbf {s}^{j(i)})^{T}(\mathbf{s}^{{j}}-\mathbf{s}^{{j(i)}})$ and $\nabla_{\mathbf{\bar{p}}^{j}_{k}}{V}_{k}(\mathbf{\bar{p}}^{j(i)}_{k})^{T}(\mathbf{\bar{p}}^{j}_{k}-\mathbf{\bar{p}}^{j(i)}_{k})$ are given as follows:  
%\begin{align}&\hspace{-0.15cm}
%\nabla_{\mathbf{s}^{j}}H^{j}(\mathbf{s}^{j({i})})^{T}(\mathbf{s}^{j}-\mathbf{s}^{j({i})})  =\sum_{k=1}^{K}\sum_{m^{j}=1}^{M^{j}}\sum_{n^{j}=1}^{N^{j}}2s^{j({i})}_{k}[m^{j},n^{j}]\left( s^{j}_{k}[m^{j},n^{j}]-s^{j(i)}_{k}[m^{j},n^{j}]\right),\forall j,\\&
%\nabla_{\boldsymbol{\alpha}}H(\boldsymbol{\alpha}^{({i})})^{T}(\boldsymbol{\alpha}-\boldsymbol{\alpha}^{({i})})   =\sum_{k=1}^{K}2\alpha^{({i})}_{k}\left( \alpha_{k}-\alpha^{({i})}_{k}\right),\\&\vspace{-0.5cm} \nabla_{\mathbf{\bar{p}}_{k}^{j}}\bar{V}^{j}_{k}(\mathbf{\bar{p}}_{k}^{j({i})})
%= \frac{a
%	Q^{-1}(\epsilon^{j}_{k})}{\sqrt{\sum_{m^{j}=1}^{M^{j}}\sum_{n^{j}=1}^{N^{j}} \bar{V}^{j(i)}_{k}[m^{j},n^{j}]}}\begin{pmatrix} 
%\ \frac{g^{j}_{k}[1]}{(1+\bar{p}^{j(i)}_{k}[1,1]g^{j}_{k}[1])^{3}} \\
%\vdots \\
%\frac{g_{k}^{j}[M]}{(1+\bar{p}_{k}^{j(i)}[M,N]g_{k}^{j}[M])^{3}} \end{pmatrix}, \forall j, 
%\end{align} 
 By substituting the right hand sides of (\ref{htayslorhu})-(\ref{halph}) into (\ref{Op3}), we obtain the following optimization problem:    \vspace{-0.5cm}
\begin{align}&\label{op3a}
\underset {\boldsymbol{f},\mathbf{s}^{u}, \mathbf{p}^{u},\mathbf{s}^{d}, \mathbf{p}^{d},\mathbf{\bar{p}}^{u}, \mathbf{\bar{p}}^{d},  \boldsymbol{\alpha}}{ \mathop {\mathrm {minimize}}\nolimits }~\bar{\bar{\Phi}}+\bar{\beta}(\mathbf{s}^{u},\mathbf{s}^{u(i)},\mathbf{s}^{d},\mathbf{s}^{d(i)},\boldsymbol{\alpha},\boldsymbol{\alpha}^{(i)})	
\\&\;\; \mathrm {s.t.}~\overline{\overline{\mathrm{C1}}}:	F^{u}_{k}(\mathbf{\bar{p}}_{k}^{u})-\bar{V}^{u}_{k}(\mathbf{\bar{p}}_{k}^{u}, \mathbf{\bar{p}}^{u(i)}_{k})\geq (1-\alpha_{k}) B_{k},\forall k, \nonumber\\& \;\qquad \overline{\overline{\mathrm{C2}}}:		F^{d}_{k}(\mathbf{\bar{p}}_{k}^{d})-\bar{V}^{d}_{k}(\mathbf{\bar{p}}_{k}^{d}, \mathbf{\bar{p}}^{d(i)}_{k})\geq (1-\alpha_{k})\Gamma_{k}B_{k} \nonumber,\forall k, \;\;\mathrm{C3-C23},\nonumber
\end{align}  
where 
$
\bar{\beta}(\mathbf{s}^{u},\mathbf{s}^{u(i)},\mathbf{s}^{d},\mathbf{s}^{d(i)},\boldsymbol{\alpha},\boldsymbol{\alpha}^{(i)})=\eta_{1}(E^{u}(\mathbf{s}^{{u}})-\bar{H}^{u}(\mathbf{s}^{{u}},\mathbf{s}^{{u(i)}}))+\eta_{2}(E^{d}(\mathbf{s}^{{d}})-\bar{H}^{d}(\mathbf{s}^{{d}},\mathbf{s}^{{d(i)}}))+\eta_{3}({E}^{\alpha}(\boldsymbol{\alpha})-\bar{H}^{\alpha}(\boldsymbol{\alpha},\boldsymbol{\alpha}^{(i)})).$
Optimization problem (\ref{op3a}) is a convex optimization problem. To facilitate the application of CVX for solving problem (\ref{op3a}), we reformulate the cubic function $f_{k}^{3}$ appearing in the cost function and transform it into two equivalent SOC constraints\cite{cvx}. We first define new auxiliary variables $\bar{\zeta}_{k}, \forall k,$ to upper bound the cubic function  as follows $f_{k}^{3}\leq\bar{\zeta}_{k}, \forall k$. Then, as shown in \cite{cvx}, we can expand $f_{k}^{3}\leq\bar{\zeta}_{k}, \forall k,$ to the following  equivalent SOC constraints\cite{cvx}: 
\begin{IEEEeqnarray}{lll}\label{cone1} 
	\mathrm{C24}:	\begin{bmatrix}
		\bar{\zeta}_{k} & \bar{\theta}_{k} \\
		\bar{\theta}_{k} & f_{k}
	\end{bmatrix}\succeq	 0,\forall k, \qquad
\mathrm{C25}: 	\begin{bmatrix}
	\bar{\theta}_{k} & f_{k} \\
	f_{k} & 1
\end{bmatrix}\succeq	 0, \forall k,
\end{IEEEeqnarray} 
where $\bar{\theta}_{k}, \forall k,$ are new auxiliary variables.
Optimization problem  (\ref{op3a}) is transformed into the following equivalent form:\vspace{-0.75cm}
\begin{align}&\label{opf}
\underset {\boldsymbol{f},\mathbf{s}^{u}, \mathbf{p}^{u},\mathbf{s}^{d}, \mathbf{p}^{d},\mathbf{\bar{p}}^{u}, \mathbf{\bar{p}}^{d},  \boldsymbol{\alpha}, \boldsymbol{\bar{\zeta}}, \boldsymbol{\bar{\theta}}}{ \mathop {\mathrm {minimize}}\nolimits }~\bar{\bar{\bar{\Phi}}}+\bar{\beta}(\mathbf{s}^{u},\mathbf{s}^{u(i)},\mathbf{s}^{d},\mathbf{s}^{d(i)},\boldsymbol{\alpha},\boldsymbol{\alpha}^{(i)})	
\\&\;\; \mathrm {s.t.}~\overline{\overline{\mathrm{C1}}}, \overline{\overline{\mathrm{C2}}}, \mathrm{C3-C25}, \nonumber
\end{align}  
where \vspace{-0.5cm}
\begin{align}&\label{obj2}
\bar{\bar{\bar{\Phi}}}=\sum_{k=1}^{K}w_{k}\kappa \bar{\zeta}_{k}+\sum_{k=1}^{K}w_{k}\bigg(\delta_{k}\sum_{m^{u}=1}^{M^{u}}\sum_{n^{u}=1}^{N^{u}}\bar{p}^{u}_{k}[m^{u},n^{u}]+(1-\alpha_{k})P_{\text{cir}}\bigg)+\delta_{\textrm{BS}}\sum_{k=1}^{K}\sum_{m^{d}=1}^{M^{d}}\sum_{n^{d}=1}^{N^{d}}\bar{p}^{d}_{k}[m^{d},n^{d}],
\end{align} 
and $\boldsymbol{\bar{\zeta}}$ and $\boldsymbol{\bar{\theta}}$ are the collection of optimization variables $\bar{\zeta}_{k}, \forall k,$ and $\bar{\theta}_{k}, \forall k$, respectively. Optimization problem (\ref{opf}) is convex because the objective function is convex and can be efficiently solved by standard convex optimization solvers such as CVX \cite{cvx}. \textbf{Algorithm} \ref{sca1} summarizes the main steps for solving (\ref{Op3}) in an iterative manner, where the
solution of (\ref{opf}) in iteration ($i$) is used as the initial point for the next iteration $(i+1)$. By iteratively solving (\ref{opf}), \textbf{Algorithm} \ref{sca1} produces a sequence of
improved feasible solutions, which for sufficiently large  $I_{\text{max}}$  convergence to a local optimum
point of problem (\ref{Op3}) or equivalently problem (\ref{Op1}) in polynomial time, \cite{pccp,fastglobal}.
\begin{algorithm}[t]	
	\begin{algorithmic}[1]
	\STATE {Initialize:} Random initial points $\mathbf{s}^{u(1)}$, $\mathbf{s}^{d(1)}$, $\mathbf{\bar{p}}^{u(1)}$, $\mathbf{\bar{p}}^{d(1)}$, $\boldsymbol{\alpha}^{(1)}$. Set iteration index $i=1$, maximum number of iterations $I_{\text{max}}$, and penalty factors $\eta_{1} >0$, $\eta_{2} >0$, and $\eta_{3} >0$.\\
	\STATE \textbf{Repeat}\\
	\STATE Solve convex problem (\ref{opf}) for given  $\mathbf{s}^{u(i)}$, $\mathbf{s}^{d(i)}$, $\mathbf{\bar{p}}^{u(i)}$, $\mathbf{\bar{p}}^{d(i)}$, $\boldsymbol{\alpha}^{(i)}$, and store the intermediate solutions $\mathbf{s}^{u}$, $\mathbf{s}^{d}$, $\mathbf{\bar{p}}^{u}$, $\mathbf{\bar{p}}^{d}$, $\boldsymbol{\alpha}$\\
	\STATE Set ${i}={i}+1$ and update $\mathbf{s}^{u(i)}=\mathbf{s}^{u}$, $\mathbf{s}^{d(i)}=\mathbf{s}^{d}$,  $\mathbf{\bar{p}}^{u(i)}=\mathbf{\bar{p}}^{u}$,
	$\mathbf{\bar{p}}^{d(i)}=\mathbf{\bar{p}}^{d}$, $\boldsymbol{\alpha}^{(i)}=\boldsymbol{\alpha}$. \\
	\STATE \textbf{Until} $i=I_{\text{max}}$.\\
	\STATE {Output:} $\mathbf{s}^{u*}=\mathbf{s}^{u}$,
	$\mathbf{s}^{d*}=\mathbf{s}^{d}$,
	$\mathbf{\bar{p}}^{u*}=\mathbf{\bar{p}}^{u}$,
	$\mathbf{\bar{p}}^{d*}=\mathbf{\bar{p}}^{d}$, $\boldsymbol{\alpha}^{*}=\boldsymbol{\alpha}$.
\end{algorithmic}
\caption{Successive Convex Approximation}
\label{sca1}
\end{algorithm}    
\subsection{Proposed SCA-Based Suboptimal Scheme 2}
For suboptimal scheme 1, we have adopted the Big-M method to linearize non-convex product terms. However, this method introduced additional optimization variables and constraints, which negatively affect the complexity of \textbf{Algorithm} \ref{sca1}. In this subsection, we reduce the complexity of  suboptimal scheme 1 (\textbf{Algorithm} \ref{sca1}). To do so, we first approximate the dispersion in the high SNR regime as follows: 
\begin{align}\label{vda}
\widetilde{\nu}[i]=\bigg(1-\frac{1}{(1+\gamma[i])^2}\bigg) \approx 1,
\end{align} 
which is accurate when the received SNR $\gamma[i]$, exceeds $5$ dB as is typically the case in cellular networks, especially when supporting URLLC\cite{cShe,csunoptimizing,chsecross}. On the other hand, in the low SNR regime, by substituting ${\nu}[i]=\widetilde{\nu}[i] = 1$ in (\ref{normalapproximation}), we obtain a lower bound on the achievable rate. If the lower bound is used for optimization of the resource allocation in MEC systems, the feasibility of the obtained solution is guaranteed. Hence, exploiting this approximation, we rewrite the dispersion parts  for the uplink and downlink in optimization problem (\ref{Op1}) as follows:
\begin{align}\label{b3c}
\tilde{V}^{j}_{k}(\mathbf{s}^{j}_{k})=aQ^{-1}(\epsilon^{j}_{k})\sqrt{\sum_{m^{j}=1}^{M^{j}}\sum_{n^{j}=1}^{N^{j}}s^{j}_{k}[m^{j},n^{j}]}, \forall j=\{u,d\}.
\end{align} 
Now, defining  
$\tilde{p}^{j}_{k}[m^{j},n^{j}]=s^{j}_{k}[m^{j},n^{j}]p^{j}_{k}[m^{j},n^{j}], \forall k,m^{j},n^{j}, \forall j\in\{u,d\},$ as new optimization variables, and rewriting $\tilde{V}^{j}_{k}(\mathbf{s}^{j}_{k},\mathbf{p}^{j}_{k})$  in (\ref{b3c}) as $\tilde{V}^{j}_{k}(\mathbf{s}^{j}_{k})$,  optimization problem  (\ref{Op1}) can be transformed as follows:
\begin{align}\label{oaa}
&\underset {\boldsymbol{f}, \mathbf{s}^{u}, \mathbf{\tilde{p}}^{u}, \mathbf{s}^{d}, \mathbf{\tilde{p}}^{d}, \boldsymbol{\alpha}}{ \mathop {\mathrm {minimize}}\nolimits }~	{\tilde{\Phi}}
\\&\;\; \mathrm {s.t.}~ \widetilde{\mathrm{C1}}:	\widetilde{F}^{u}_{k}(\mathbf{s}_{k}^{u},\mathbf{\tilde{p}}_{k}^{u})-\widetilde{V}^{u}_{k}(\mathbf{s}_{k}^{u})\geq (1-\alpha_{k})B_{k},\forall k, \quad \widetilde{\mathrm{C2}}:\widetilde{F}^{d}_{k}(\mathbf{s}_{k}^{d},\mathbf{\tilde{p}}_{k}^{d})-\widetilde{V}^{d}_{k}(\mathbf{s}_{k}^{d})\geq (1-\alpha_{k})\Gamma_{k}B_{k} \nonumber,\forall k, \\&  \;\qquad	\mathrm{C3-C8},
\widetilde{\mathrm{C9}}: \sum_{m^{u}=1}^{M^{u}}\sum_{n^{u}=1}^{N^{u}}\tilde{p}^{u}_{k}[m^{u},n^{u}] \leq P_{k,\text{max}}, \forall k,
\widetilde{\mathrm{C10}}: \tilde{p}^{u}_{k}[m^{u},n^{u}] \geq 0, \forall k,m^{u}, n^{u},\nonumber\\& \;\qquad \widetilde{\mathrm{C11}}: \sum_{k=1}^{K}\sum_{m^{d}=1}^{M^{d}}\sum_{n^{d}=1}^{N^{d}}\tilde{p}^{d}_{k}[m^{d},n^{d}] \leq P_{\text{max}},\;\; \widetilde{\mathrm{C12}}:\; \tilde{p}^{d}_{k}[m^{d},n^{d}] \geq 0, \forall k, m^{d}, n^{d}, \quad \mathrm{C13-C15}.\nonumber
\end{align} 
where
\begin{align}&\label{obj3}
{\tilde{\Phi}}=	\sum_{k=1}^{K}w_{k}\big(\kappa f_{k}^{3}+\delta_{k}\sum_{m^{u}=1}^{M^{u}}\sum_{n^{u}=1}^{N^{u}}\tilde{p}^{u}_{k}[m^{u},n^{u}]+(1-\alpha_{k})P_{\text{cir}}\big)+\delta_{\textrm{BS}}\sum_{k=1}^{K}\sum_{m^{u}=1}^{M^{u}}\sum_{n^{d}=1}^{N^{d}}\tilde{p}^{d}_{k}[m^{u},n^{d}],
\end{align} 
$\mathbf{\tilde{p}}^{j}_{k}, \forall j,$ are the collection of optimization variables $\tilde{p}^{j}_{k}[m^{j},n^{j}],\forall k,m^{j},n^{j}, \forall j$, $\mathbf{\tilde{p}}^{j}, \forall j,$ denote the collection of optimization variables $\mathbf{\tilde{p}}^{j}_{k}, \forall k, \forall j$, and   
\begin{IEEEeqnarray}{lll}\label{fr}
	\widetilde{F}^{j}_{k}(\mathbf{s}_{k}^{j},\tilde{\mathbf{p}}_{k}^{j})=\sum_{m^{j}=1}^{M^{j}}\sum_{n^{j}=1}^{N^{j }}s^{j}_{k}[m^{j},n^{j}]\log_{2}\bigg(1+\frac{g^{j}_{k}[m^{j}]\tilde{p}^{j}_{k}[m^{j},n^{j}]}{s^{j}_{k}[m^{j},n^{j}]}\bigg), \forall j.
\end{IEEEeqnarray} 
Although $\widetilde{F}^{j}_{k}(\mathbf{s}_{k}^{j},\tilde{\mathbf{p}}_{k}^{j})$ is a concave function, optimization problem (\ref{oaa}) is not convex due to the non-convexity of constraints $\widetilde{\mathrm{C1}}$, $\widetilde{\mathrm{C2}}$, $\mathrm{C6}$, $\mathrm{C10}$, and $\mathrm{C14}$. To deal with non-convex constraints $\widetilde{\mathrm{C1}}$ and $\widetilde{\mathrm{C2}}$, we define new optimization variables $z_{k}, \forall k,$ and $q_{k}, \forall k$, and rewrite the constraint equivalently as follows:
\begin{align}
&\hspace{-0.35cm}\widetilde{\mathrm{C1a}}:\;F^{u}_{k}(\mathbf{s}_{k}^{u},\mathbf{\tilde{p}}_{k}^{u})-aQ^{-1}(\epsilon^{u}_{k})z_{k}\geq (1-\alpha_{k})B_{k},\forall k, \quad
\widetilde{\mathrm{C1b}}:z_{k} \geq \sqrt{\sum_{m^{u}=1}^{M^{u}}\sum_{n^{u}=1}^{N^{u}}(s^{u}_{k}[m^{u},n^{u}]})^{2},\\&\hspace{-0.35cm}
\widetilde{\mathrm{C2a}}:F^{d}_{k}(\mathbf{s}_{k}^{d},\mathbf{\tilde{p}}_{k}^{d})-aQ^{-1}(\epsilon^{d}_{k})q_{k}\geq (1-\alpha_{k})\Gamma_{k}B_{k},\forall k,\quad
\widetilde{\mathrm{C2b}}:q_{k} \geq \sqrt{\sum_{m^{d}=1}^{M^{d}}\sum_{n^{d}=1}^{N^{d}}(s^{d}_{k}[m^{d},n^{d}]})^{2}.
\end{align} 
Constraints $\widetilde{\mathrm{C1b}}$ and $\widetilde{\mathrm{C2b}}$ are rewritten in this form as for the optimal solution $s^{j}_{k}[m^{j},n^{j}]=(s^{j}_{k}[m^{j},n^{j}])^{2}$ holds. Constraints $\widetilde{\mathrm{C1a}}$, $\widetilde{\mathrm{C1b}}$, $\widetilde{\mathrm{C2b}}$, and $\widetilde{\mathrm{C2b}}$ span a convex set since constraints $\widetilde{\mathrm{C1b}}$ and $\widetilde{\mathrm{C2b}}$ can be represented as SOCs. To deal with constraints $\mathrm{C6}$, $\mathrm{C8}$, and $\mathrm{C14}$ and the cubic function present in optimization problem (\ref{oaa}), we use similar techniques as in suboptimal scheme 1. As a consequence, optimization problem (\ref{oaa}) is rewritten in the following equivalent form:  
\begin{align}\label{oaa2}
& \underset { \boldsymbol{f}, \mathbf{s}^{u}, \mathbf{\tilde{p}}^{u}, \mathbf{s}^{d}, \mathbf{\tilde{p}}^{d}, \boldsymbol{\alpha},\mathbf{z}, \mathbf{q},\boldsymbol{\bar{\zeta}}, \boldsymbol{\bar{\theta}}}{ \mathop {\mathrm {minimize}}\nolimits }~	{\tilde{\tilde{\Phi}}}+\nonumber\bar{\beta}(\mathbf{s}^{u},\mathbf{s}^{u(i)},\mathbf{s}^{d},\mathbf{s}^{d(i)},\boldsymbol{\alpha},\boldsymbol{\alpha}^{(i)})	
\\&\;\; \mathrm {s.t.}~\widetilde{\mathrm{C1a}},\widetilde{\mathrm{C1b}}, \widetilde{\mathrm{C2a}}, \widetilde{\mathrm{C2b}}, \mathrm{C3-C5}, \;
\mathrm{C6a}: s^{u}_{k}[m^{u},n^{u}] \in [0,1], \forall k,m^{u},n^{u}, \; \mathrm{C7},\\&  \; \qquad
 \mathrm{C8a}: s^{d}_{k}[m^{d},n^{d}] \in [0,1], \forall k,m^{d},n^{d}, \; \widetilde{\mathrm{C9}}-\mathrm{C13}, \;\nonumber
\\&  \; \qquad   	\mathrm{C14a}: \alpha_{k} \in [0,1], \forall k\nonumber,\mathrm{C15},\mathrm{C16}: 
\begin{bmatrix}
\bar{\theta}_{k} & f_{k} \\
f_{k} & 1
\end{bmatrix}\succeq	 0, \forall k, \quad
\mathrm{C17}:	\begin{bmatrix}
\bar{\zeta}_{k} & \bar{\theta}_{k} \\
\bar{\theta}_{k} & f_{k}
\end{bmatrix}\succeq	 0, \forall k,
\end{align}
where\vspace{-1cm}
\begin{align}&\label{obj}
{\tilde{\tilde{\Phi}}}=	\sum_{k=1}^{K}w_{k}\bigg(\kappa\zeta_{k} +\delta_{k}\sum_{m^{u}=1}^{M^{u}}\sum_{n^{u}=1}^{N^{u}}\tilde{p}^{u}_{k}[m^{u},n^{u}]+(1-\alpha_{k})P_{\text{cir}}\bigg)+\delta_{\textrm{BS}}\sum_{k=1}^{K}\sum_{m^{u}=1}^{M^{u}}\sum_{n^{d}=1}^{N^{d}}\tilde{p}^{d}_{k}[m^{u},n^{d}],
\end{align}  
and $\mathbf{z}$ and $\mathbf{q}$ are the collection of optimization variables $z_{k}, \forall k,$ and $q_{k}, \forall k,$ respectively. Optimization problem (\ref{oaa2}) is convex because the objective function is convex and the constraints span a convex set. Therefore, it can be efficiently solved by standard convex optimization solvers such as CVX \cite{cvx}. \textbf{Algorithm} \ref{sca2} summarizes the main steps for solving (\ref{oaa}) in an iterative manner, where the
solution of (\ref{oaa2}) in iteration ($i$) is used as the initial point for
the next iteration $(i+1)$. The algorithm produces a sequence of
improved feasible solutions until convergence to a local optimum
point of problem (\ref{oaa}). Unlike \textbf{Algorithm} 2, \textbf{Algorithm} \ref{sca2} does not provide a local optimum solution to problem (\ref{Op1}) because of the approximation of the dispersion term. Nevertheless, \textbf{Algorithm} \ref{sca2} provides an upper bound on the total system power consumption and the obtained solution is feasible for (\ref{Op1}). Moreover, this upper bound becomes tight for sufficiently high SNR, where the approximation in (\ref{vda}) becomes tight, which is likely the case for URLLC applications.    
 \subsection{Complexity Analysis of Suboptimal Algorithms}
In this sub-section, we study the complexity of the proposed low-complexity suboptimal schemes. 

\textit{1) Suboptimal Algorithm 1:} Optimization problem (\ref{opf}) is a non-linear convex problem which can be solved efficiently in polynomial time using e.g., CVX \cite{cvx}. There are in total $3KM^{u}N^{u}+3KM^{d}N^{d}+4K$ optimization variables and $M^{u}N^{u}(2+7K)+M^{d}N^{d}(2+7K)+8K+KM^{u}M^{d}\bar{O}$ linear and convex constraints. Thus, the computational complexity order of Algorithm \ref{sca1} per iteration is $\mathcal{O}\big((3KM^{u}N^{u}+3KM^{d}N^{d}+4K)^{3}(M^{u}N^{u}(M^{u}N^{u}(2+7K)+M^{d}N^{d}(2+7K)+8K+KM^{u}M^{d}\bar{O})\big)$\cite{ataj2,Ben1,Pólik2010}.

\textit{2) Suboptimal Algorithm 2:} Optimization problem (\ref{oaa2}) is also a non-linear convex problem, which can be solved efficiently in polynomial time using e.g., CVX \cite{cvx}. There are in total $2KM^{u}N^{u}+2KM^{d}N^{d}+6K$ optimization variables and $M^{u}N^{u}(2+2K)+M^{d}N^{d}(2+2K)+8K+KM^{u}M^{d}\bar{O}$ linear and convex constraints. Thus, the computational complexity order of Algorithm \ref{sca1} per iteration is $\mathcal{O}\big((2KM^{u}N^{u}+2KM^{d}N^{d}+6K)^{3}(M^{u}N^{u}(2+2K)+M^{d}N^{d}(2+2K)+8K+KM^{u}M^{d}\bar{O})\big)$\cite{ataj2,Ben1,Pólik2010}. 
As can be observed, the complexity of the proposed suboptimal scheme 2 is lower than that of the proposed suboptimal scheme 1. The reason behind this is the smaller number of optimization variables and  constraints in (\ref{oaa2}) compared to (\ref{opf}). 
\begin{algorithm}[t]	
	\begin{algorithmic}[1]
	\STATE {Initialize:} Random initial points $\boldsymbol{\alpha}^{(1)},$ $\mathbf{s}^{u(1)},$ $\mathbf{s}^{d(1)}$. Set iteration index $i=1$, maximum number of iterations $I_{\text{max}}$, and penalty factors, $\eta_{1} >0$, $\eta_{2} >0$, and $\eta_{3} >0$.\\
	\STATE \textbf{Repeat}\\
	\STATE Solve convex problem (\ref{oaa2}) for given   $\boldsymbol{\alpha}^{(i)},$ $\mathbf{s}^{u(i)},$ $\mathbf{s}^{d(i)},$ and store the intermediate solutions  $\boldsymbol{\alpha}$, $\mathbf{s}^{u}$, $\mathbf{s}^{d}$.\\
	\STATE Set ${i}={i}+1$ and update  $\boldsymbol{\alpha}^{(i)}=\boldsymbol{\alpha}$, $\mathbf{s}^{u(i)}=\mathbf{s}^{u}$, $\mathbf{s}^{d(i)}=\mathbf{s}^{d}$. \\
	\STATE \textbf{Until} $i=I_{\text{max}}$.\\
	\STATE {Output:} $\boldsymbol{\alpha}^{*}=\boldsymbol{\alpha}$, $\mathbf{s}^{u*}=\mathbf{s}^{u}$,
	$\mathbf{s}^{d*}=\mathbf{s}^{d}$,
	$\mathbf{\tilde{p}}^{u*}=\mathbf{\tilde{p}}^{u}$,
	$\mathbf{\tilde{p}}^{d*}=\mathbf{\tilde{p}}^{d}$.
\end{algorithmic} 
	\caption{Successive Convex Approximation}
	\label{sca2}
\end{algorithm} 
\section{Performance Evaluation}
In this section, we provide simulation results to evaluate the performance of the proposed joint uplink-downlink resource allocation algorithm for OFDMA-URLLC MEC systems. We adopt the simulation parameters provided in Table I, unless specified otherwise. In our simulations, a single cell is considered with inner and outer radii of $r_{1}$ and $r_{2}$, respectively. The BS is located at the center of the cell, and the URLLC users are randomly located between the inner and the outer radii. The path loss is calculated as  $35.3 + 37.6 \log_{10}(d_{k})$\cite{chsecross}, where $d_{k}$ is the distance from the BS to user $k$. The values of the penalty factors used in \textbf{Algorithm} 2 and 3 are set to $\eta_{1}=10KP_{k, \text{max}}$  and $\eta_{2}=\eta_{3}=10P_{\text{max}}$. The small scale fading gains between the BS and the URLLC users are modeled as independent and identically Rayleigh distributed. All simulation results are averaged over $100$ realizations of the path loss and multipath fading.
\begin{table}[t]\vspace{-0.5cm}
	\label{tab:table}
	\centering
	\caption{Simulation Parameters Settings.} 
	\renewcommand{\arraystretch}{1.4}
	\scalebox{0.6}{%
		\begin{tabular}{|c||c|c||c|}
			\hline
			Parameter & Value & Parameter & Value\\ \hline \hline 
			Total number  of sub-carriers in uplink and downlink $M=M^{u}=M^{d}$ & $M_{T}=2M$=64 	&		Number  of time slots in uplink and downlink $N^{u}=N^{d}$ & 4\\ \hline
			Bandwidth of each sub-carrier & 30 kHz & Noise power density  & -174 dBm/Hz  \\ \hline
			Maximum BS transmit power, $P_{\text{max}}$  &  $45$~dBm & Maximum transmit power of each user, $P_{k, \text{max}}$  &  $25$~dBm \\ \hline 
		Packet error probability 	& $\epsilon^{j}_{k}=10^{-6}, \forall j, k$& Circuit power consumption of user $k$, $P_{\text{cir}}$& $50$~\textrm{mW}\cite{user_centric}\\ \hline
			Value of $\Gamma_{k},\forall k$ &  1	& Effective switched capacitance $\kappa$&  $10^{-27}$ Farad 					\\ \hline 			
			Required CPU cycles for processing one bit of 
			information $c_{k}$ &  $[100-5000]$ cycles/bit \cite{miettinen2010energy} & 			Maximum CPU frequency of local user processor $f_{\text{max}}$&  $2.7$ GHz				\\  \hline
			Power amplifier inefficiency of the users and the BS &  $\delta_{k}=1, \forall k,$ and $\delta_{\textrm{BS}}=1$ &  Users weights &  $w_{k}=1,\forall k$  				\\  \hline    			
	\end{tabular}}
\end{table} 
%\begin{table}[t]
%	\label{tab:table}
%	\centering
%	\caption{Simulation Parameters Settings.} 
%	\renewcommand{\arraystretch}{1.4}
%	\scalebox{0.5}{%
%		\begin{tabular}{|c||c|}
%			\hline
%			Parameter & Value \\ \hline \hline 
%			Total number  of sub-carriers in uplink and downlink $M=M^{u}=M^{d}$ & $M_{T}=2M$=64 \\ \hline
%			Number  of time slots in uplink and downlink $N^{u}=N^{d}$ & 4 \\ \hline
%		    Bandwidth of each sub-carrier & 30 kHz \\ \hline
%			Noise power density  & -174 dBm/Hz \\ \hline
%			Maximum BS transmit power, $P_{\text{max}}$  &  $45$~dBm \\ \hline  
%		Maximum transmit power of each user, $P_{k, \text{max}}$  &  $25$~dBm \\ \hline  
%	    Dynamic power consumption of the BS& $0.1$~\textrm{mW}\cite{powermodel}\\ \hline
%	  Static circuit power consumption of user $k$& $50$~\textrm{mW}\cite{powermodel}\\ \hline
%			Value of $\Gamma_{k},\forall k$ &  1						\\ \hline 	Effective switched capacitance $\kappa$&  $10^{-27}$						\\ \hline 			
%			Required CPU cycles for processing one bit of 
%			information $c_{k}$ &  $[100-5000]$ cycles/bit \cite{miettinen2010energy}						\\ \hline
%			Maximum CPU frequency of local user processor $f_{k}$&  $2.7$ GHz						\\  	 \hline  			
%			Packet error probability 
%			 &  $\epsilon^{j}_{k}=10^{-6}, \forall j, k$ 					\\ \hline 			
%	\end{tabular}}
%\end{table} 
\subsection{Performance Bound and Benchmark Schemes}
We compare the performance of the proposed resource allocation algorithms with the following schemes:
\begin{itemize}
	\setlength{\itemsep}{1pt}
\item {\textbf{Shannon's capacity (SC)}}: To obtain an (unachievable) lower bound on the total network power consumption, Shannon's capacity formula is adopted in problem (\ref{Op1}), i.e., $V^{j}_{k}(\mathbf{s}_{k}^{j},\mathbf{p}_{k}^{j}), \forall j,$ is set to zero in constraints $\mbox {C1}$ and $\mbox {C2}$, respectively, and all other constraints are retained. The resulting optimization problem is solved using a modified version of \textbf{Algorithm} 2.
		\setlength{\itemsep}{1pt}
	\item {\textbf{Local computation (LC)}}: In this scheme, only local computation is employed where each user aims to minimize its local computation power by optimizing its own CPU frequency  subject to its delay constraint. The resulting optimization problem is convex and can be solved optimally using convex optimization tools such as CVX\cite{cvx}. 
\item {\textbf{Edge Only (EO)}}: In this scheme, all URLLC users offload their data to the edge server. The resulting optimization problem is solved using the SCA based algorithm from the conference version \cite{gha4}.
\item {\textbf{Fixed sub-carrier assignment (FSA)}}: In this scheme, we fix the sub-carrier assignment for offloading and optimize the remaining degrees of freedom via SCA. We divide the total number of sub-carriers among the users such that their delay and causality constraints are met. This can be done by solving a mixed integer feasibility problem. 
\end{itemize}

\subsection{Simulation Results}
In Figs.~\ref{conv1} and \ref{conv2}, we investigate the convergence of the proposed optimal algorithm (\textbf{Algorithm} 1) and the suboptimal algorithms (\textbf{Algorithms} 2 and 3) for different numbers of sub-carriers $M^{u}$, $M^{d}$, and different numbers of users $K$ for a given channel realization. We show the total sum power consumption as a function of the number of iterations. As can be observed from Fig.~\ref{conv1}, the proposed optimal scheme converges to the global optimal solution after a finite number of iterations. In particular, the optimal scheme converges after $100000$ and $170000$ iterations for $M_{T}=24$ and $M_{T}=32$, respectively. For the proposed optimal scheme, the number of iterations required for convergence increases significantly with the number of sub-carriers since increasing the number of sub-carriers increases the dimensionality of the search space. On the other hand, the proposed suboptimal scheme 1 (\textbf{Algorithm} 2) attains a close-to-optimal performance for a much smaller number of iterations. We note that optimization problem (\ref{Op3c}) has to be solved $2^{K}$ times to find the global optimal solution, see Section IV.B. We show in Fig.~\ref{conv1} the solution for the best $\bar{\boldsymbol{\alpha}}$.     

\begin{figure*}[!tbp]
	\centering
	\begin{minipage}{0.48\textwidth}
	\hspace{-0.3cm}	
	\resizebox{1\linewidth}{!}{\psfragfig{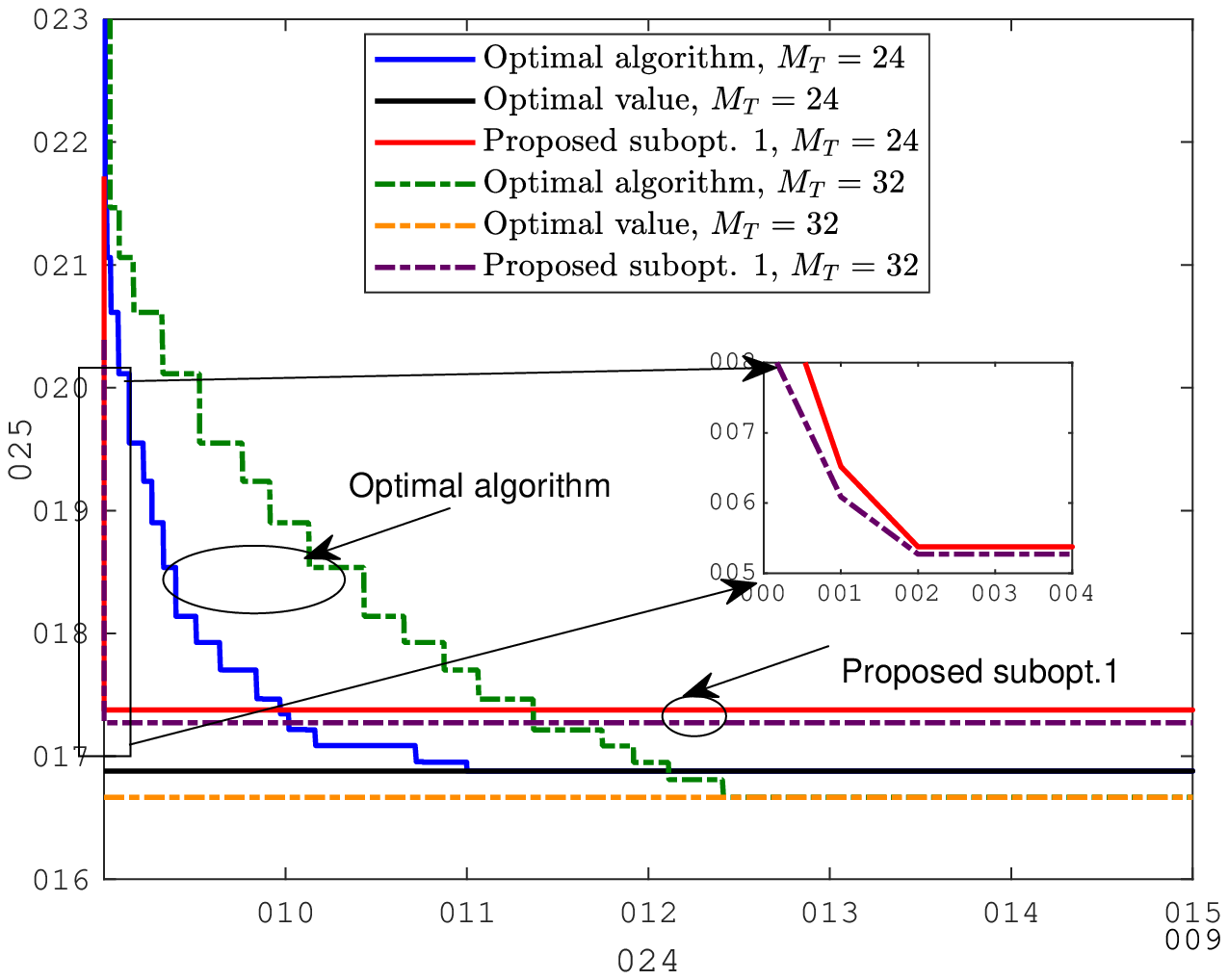}}\vspace{-4mm}
	\caption{Convergence of the proposed optimal scheme (\textbf{Algorithm} 1) and suboptimal scheme 1 (\textbf{Algorithm} 2). $M_{T}=M^{u}+M^{d},$ $K=2$, $N_{T}=N^{u}+N^{d}=2$, $B_{k}=80~\textrm{bits},\forall k$, $r_{1}=r_{2}=10~\textrm{m}$, $\tau=1$, $D_{1}=D_{2}=2$, and $c_{k}=~5000$, $\epsilon^{j}_{k}=10^{-3}, \forall j, k.$}
	\label{conv1}
\end{minipage}
	\hfill
	%%%%%%%%%%%%%%%%%%%%%%%%%%%%%%%%%%%%%%
	\begin{minipage}{0.48\textwidth}\vspace{-0.4cm}	
		\hspace{-0.3cm}	
		\resizebox{1\linewidth}{!}{\psfragfig{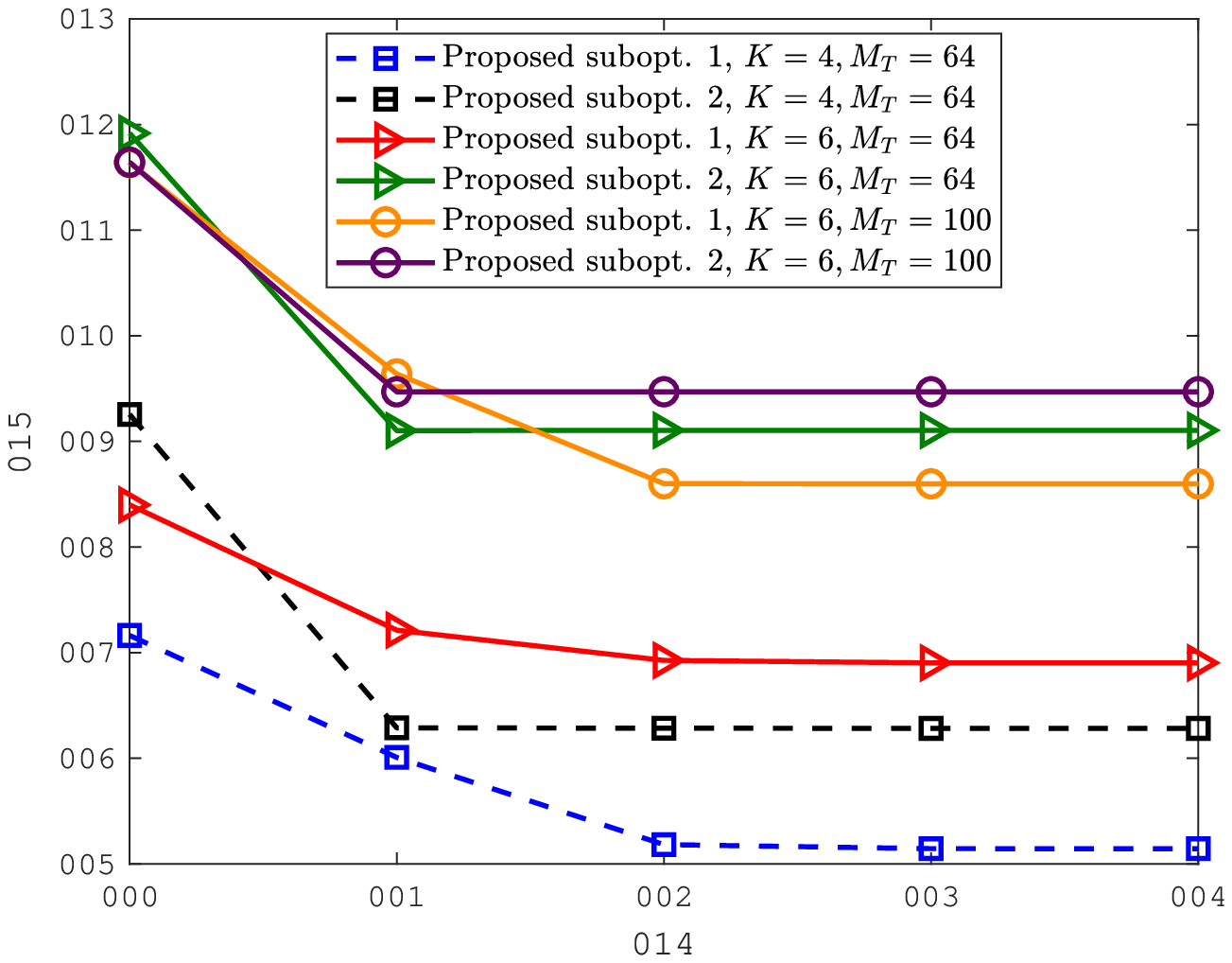}}\vspace{-4mm}
		\caption{ Convergence of the proposed suboptimal  schemes. $\tau =3,$ $B_{k}=50~\textrm{Bytes},\forall k$, $r_{1}=r_{2}=50~\textrm{m}$, $D_{1}=D_{2}=5$,  $D_{k}=7, \forall k \neq \{1,2\}$, and $c_{k}=1000, \forall k$.}
		\label{conv2}
	\end{minipage}\vspace{-0.5cm}
\end{figure*}

In Fig. \ref{conv1}, we chose relatively small values for $M^{j}, \forall j,$ $N^{j}, \forall j,$ and $K$ since the complexity of the optimal algorithm increases rapidly with the dimensionality of the problem. In Fig. \ref{conv2}, we investigate the convergence behavior of the proposed suboptimal schemes for larger values of these parameters. As can be observed from Fig. \ref{conv2}, for all considered combinations of parameter values, the proposed suboptimal schemes require a small number of iterations to converge. In particular, the proposed suboptimal scheme 1 requires at most $4$ iterations to converge while the proposed suboptimal scheme 2 requires only $2$ iterations. The reason for the faster convergence of the suboptimal scheme 2 is the convexity of the feasible set of the underlying optimization problem (\ref{oaa2}), while for suboptimal scheme 1, the feasible set of the corresponding optimization problem (\ref{opf}) is an approximated convex set, and thus, the algorithm requires more iteration to converge. On the other hand, suboptimal scheme 2 causes a higher power consumption compared to suboptimal scheme 1. The higher power consumption is caused by the approximation of channel dispersion in (\ref{vda}) used for derivation of suboptimal scheme 2  which yields an upper bound on the achievable power consumption. As expected, the convergence speeds
of the proposed suboptimal schemes are less sensitive to
the problem size and the number of users compared to that
of the optimal scheme as they avoid the costly branching operation of branch-and-bound type algorithms.

%. For the suboptimal schemes, the optimization problem is reformulated
%as DC programming problems. Then, Taylor
%series approximations are used to convexify the non-convex
%parts of the problem. By iteratively solving the approximated
%problem, a locally optimal point can be found which yields an
%objective value that is close to the globally optimum one. Moreover, as can be seen from In Fig. \ref{conv2}, the proposed suboptimal algorithms offer different trade-offs between complexity and performance.

In Figs.~\ref{bits} and \ref{bits1}, we investigate the average system power consumption versus the task size of the URLLC users. As expected, increasing the required number of computed bits leads to higher power consumption. This is due to the fact that if more bits are to be transmitted or computed in a given frame, higher SNRs or high CPU frequencies are needed, and thus, the BS and the users have to increase their powers.

 In Fig.~\ref{bits}, we compare the performance of the proposed schemes with SC. SC provides a lower bound for the required power consumption of OFDMA-URLLC MEC systems. However, SC cannot guarantee the required latency and reliability. This is due to the fact that, in this scheme, the performance loss incurred by FBT is not taken into account for resource allocation design, and thus the obtained resource allocation policies may not meet the QoS constraints.   As can be seen, the proposed suboptimal schemes attain a close-to-optimal performance. Thereby, suboptimal scheme 1 achieves a lower average system power consumption than suboptimal scheme 2 since the latter approximates the dispersion as in (\ref{vda}). On the other hand, as pointed out in Section V.C, suboptimal scheme 2 entails a low computational complexity. Hence, the proposed suboptimal schemes offer different trade-offs between performance and complexity.

 In Fig.~\ref{bits}, we chose relatively small values for $K$, $M^{u}, M^{d},  N^{u}$, and $N^{d}$ since the complexity of optimal Algorithm 1 increases
rapidly with the dimensionality of the problem, cf. Section
IV.D. In Fig.~\ref{bits1}, we investigate the performance of
the proposed suboptimal schemes for larger values of these
parameters. As can be seen, the proposed schemes lead to a substantially lower power consumption compared to the  FSA, LC, and EO schemes.  For the FSA scheme, the poor performance is due to the smaller number of degrees of freedom for resource allocation as this scheme uses a fixed sub-carrier allocation. For the LC scheme, the performance degradation is caused by the limited computation capability of the URLLC users' CPUs. Moreover, for LC scheme, the local computation is not feasible if the task size exceeds a given value. This is due to the restriction imposed by the maximum CPU frequency $f_{\text{max}}$. The proposed schemes also attain large power savings compared to the EO scheme. This is due to the joint optimization of local and edge computing, while for the EO scheme only offloading is considered.  

Moreover, as can also be seen from Fig.~\ref{bits1}, for small task sizes, suboptimal scheme 1 causes a lower power consumption than suboptimal scheme 2. This is due to the fact that for small task sizes, the users and the BS transmit with low powers leading low SNRs. In this case, the approximation in (\ref{vda}) which exploited for suboptimal scheme 2 is not accurate. On the other hand, large task sizes force the users and the BS to transmit with high power resulting in high SNRs such that the approximation becomes accurate and both suboptimal schemes have a similar performance.

%In Fig.~\ref{Pe}, we show the average system power consumption versus the packet error probability for different resource allocation schemes. As can be observed, for the proposed schemes, the FSA, and the edge only schemes the average system power consumption is a monotonically decreasing function of the packet error probability. This is due to the fact that the complementary error function in the normal approximation is a monotonically decreasing function of $\epsilon$, and as a result, the impact of the dispersion part in the normal approximation decreases as $\epsilon$ increases. Moreover, as can be seen, for the SC, the power consumption is independent of the packet error probability. This is due to the fact that SC assumes that the decoding error probability is zero. Moreover, the gap between the proposed scheme and SC is the price to be paid for enforcing strict delay and reliability requirements to ensure URLLC. Furthermore, as can be  seen, for the LC, the power consumption is also independent of the packet error probability. This is due to the fact that LC performs the tasks locally and no offloading is considered. 
\begin{figure*}[!tbp]
	\centering
	\begin{minipage}{0.48\textwidth}
		\hspace{-0.3cm}	
		\resizebox{1\linewidth}{!}{\psfragfig{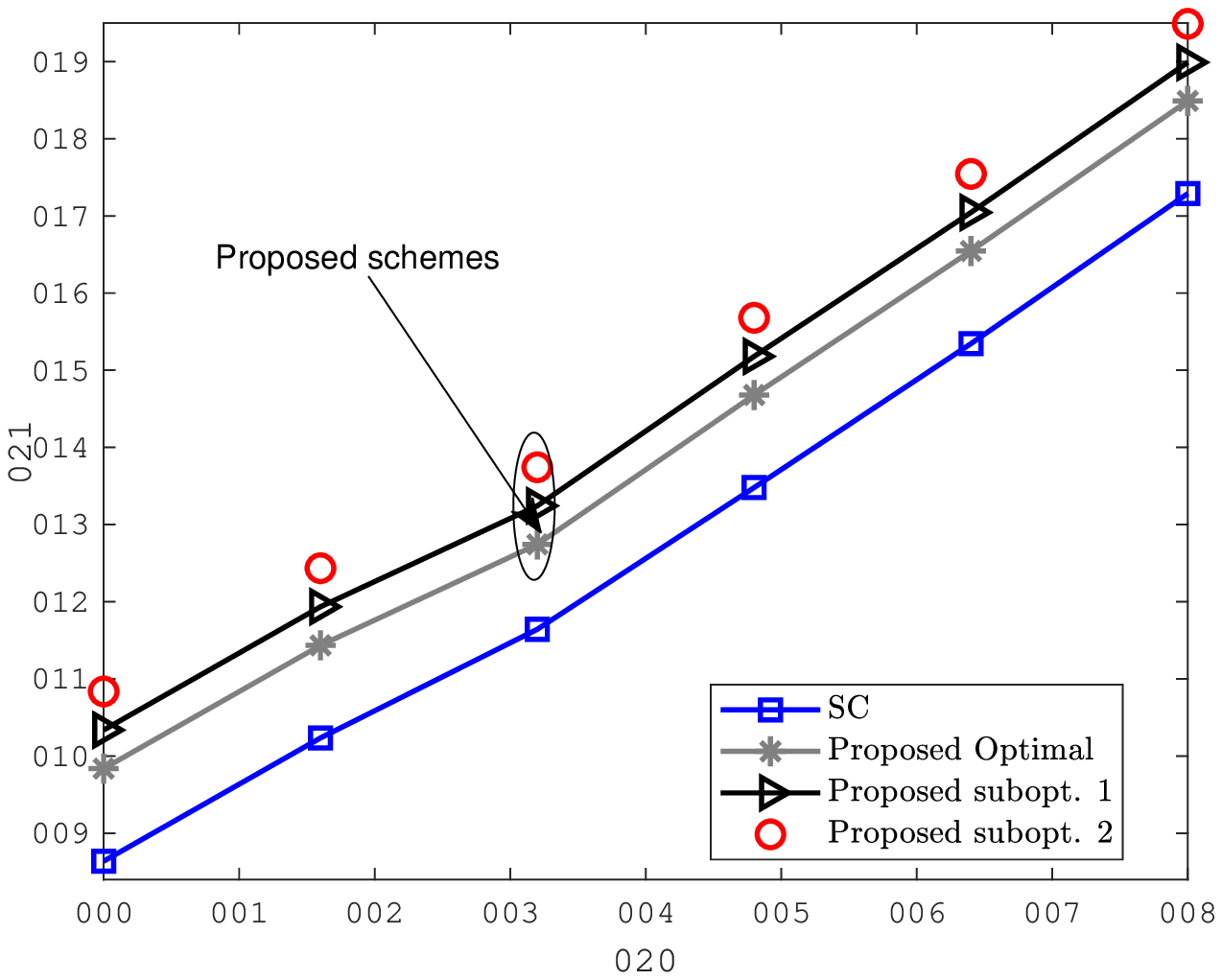}}\vspace{-4mm}
		\caption{Average system power consumption in [dBm] vs. task size in [bits] for the proposed resource allocation schemes and SC. $r_{1}=r_{2}=50~\textrm{m}$, $M^{u}=M^{d}=12$, $N^{u}=N^{d}=2$, $\tau =2,$ $K=2$, $D_{1}=D_{2}=4$, and $c_{1}~=c_{2}~=5000$.}
		\label{bits}
	\end{minipage}
	\hfill
	%%%%%%%%%%%%%%%%%%%%%%%%%%%%%%%%%%%%%%
		\begin{minipage}{0.48\textwidth}
		\hspace{-0.3cm}	
		\resizebox{1\linewidth}{!}{\psfragfig{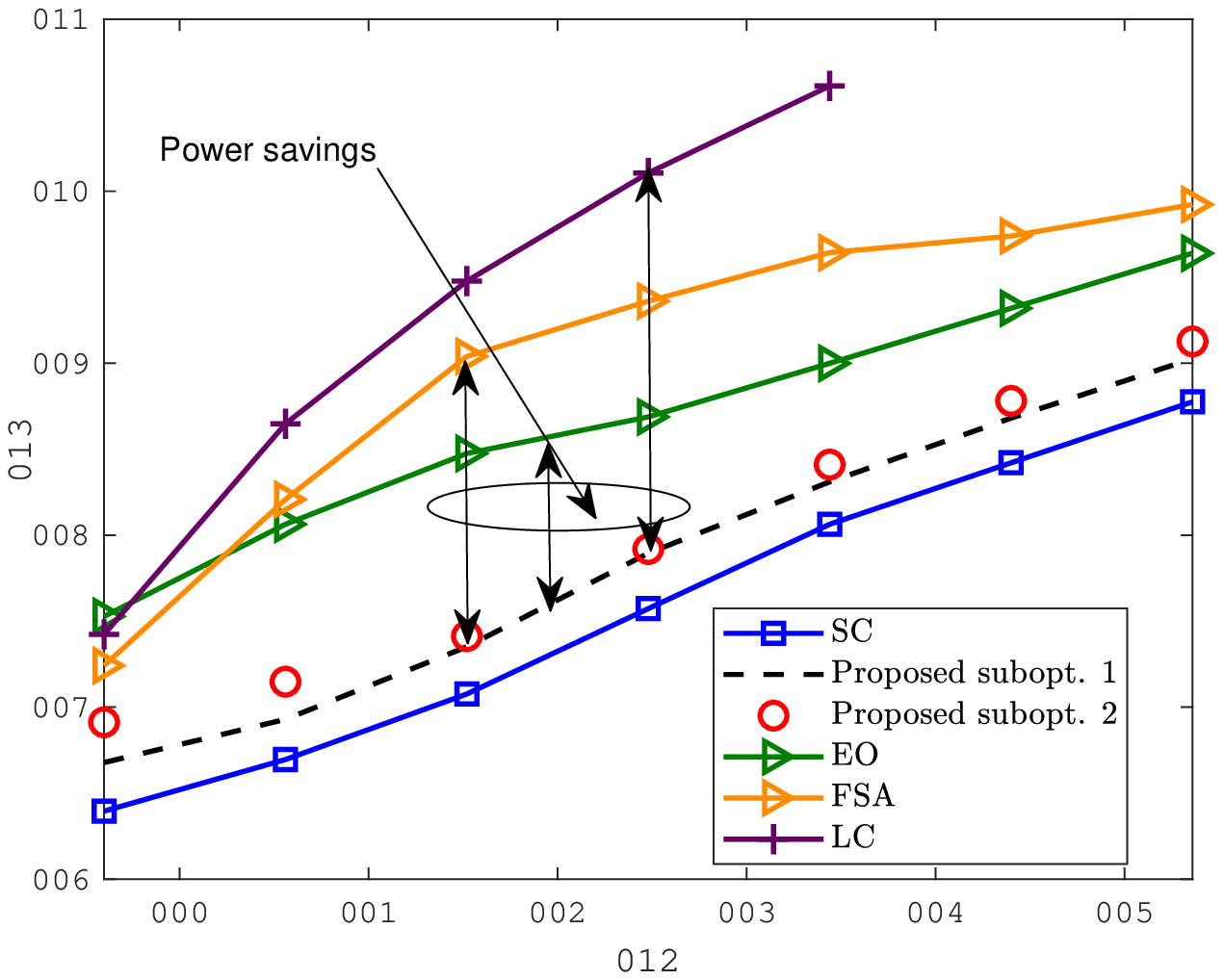}}\vspace{-4mm}
		\caption{Average system power consumption in [dBm] vs. task size in [bits] for different resource allocation schemes. $r_{1}=r_{2}=75~\textrm{m}$, $\tau =3,$ $K=4$, $D_{1}=D_{2}=5$, $D_{3}~=D_{4}~=7$, $c_{1}=c_{3}=330$, and $c_{2}~=c_{4}~=1500$.}
		\label{bits1}
	\end{minipage}\vspace{-0.75cm}
\end{figure*}
\begin{figure*}[!tbp]
	\centering	
	\begin{minipage}{0.48\textwidth}
		\hspace{-0.3cm}	
		\resizebox{1\linewidth}{!}{\psfragfig{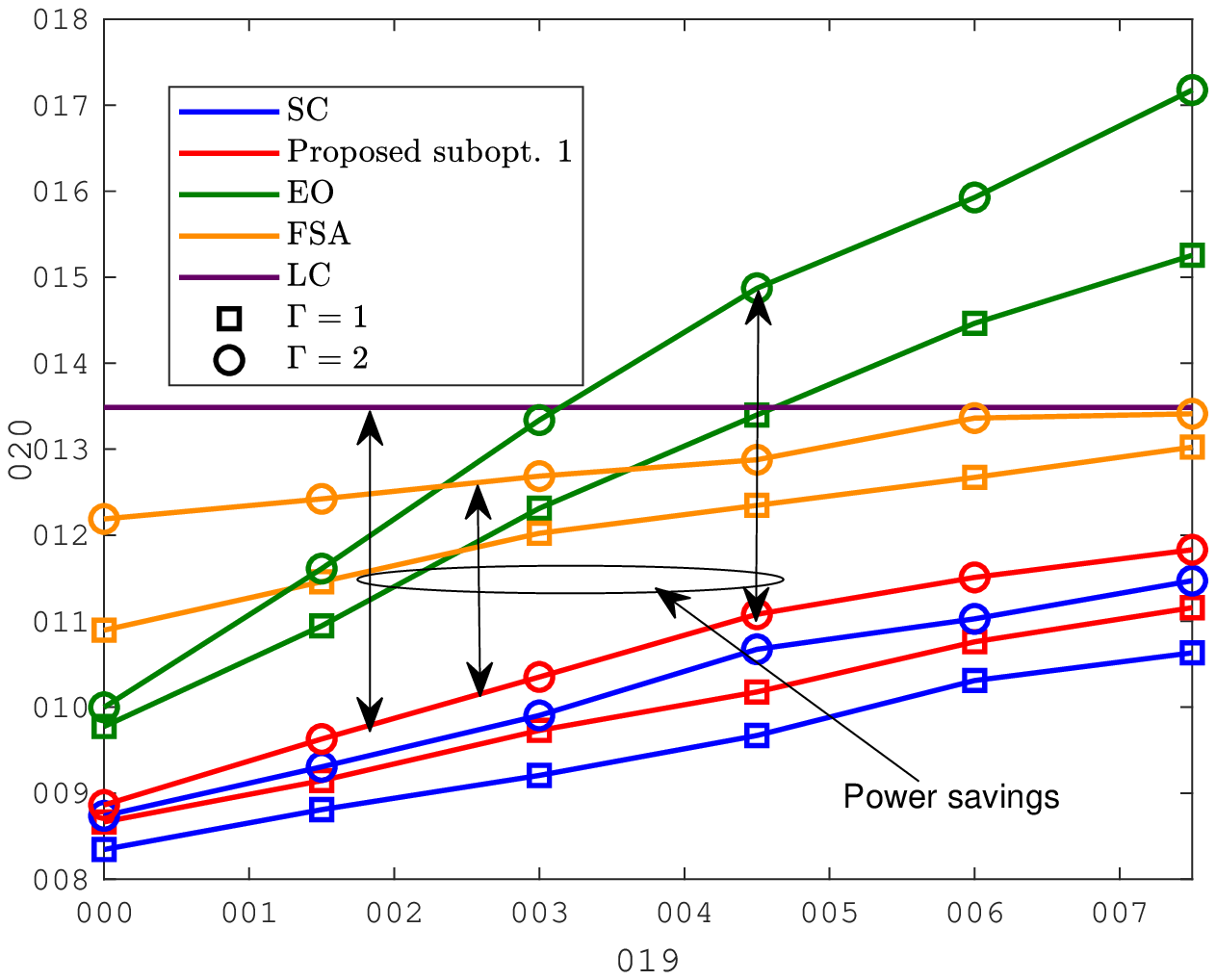}}\vspace{-4mm}
		\caption{Average system power consumption in [dBm] vs. cell size in meters
			for different resource allocation schemes. $\tau =3,$ $K=4$, $D_{1}=D_{2}=5$, and $D_{3}~=D_{4}~=7$, $B_{k}=160~\text{bits}, \forall k$, $c_{1}=c_{3}=330$, $c_{2}~=c_{4}~=1000$, and $r_{1}=20~\textrm{m}$.}
		\label{cell_raduis}
	\end{minipage}	
	\hfill
	\begin{minipage}{0.48\textwidth}	\vspace{-0.6cm}	
		\hspace{-0.3cm}	
		\resizebox{1\linewidth}{!}{\psfragfig{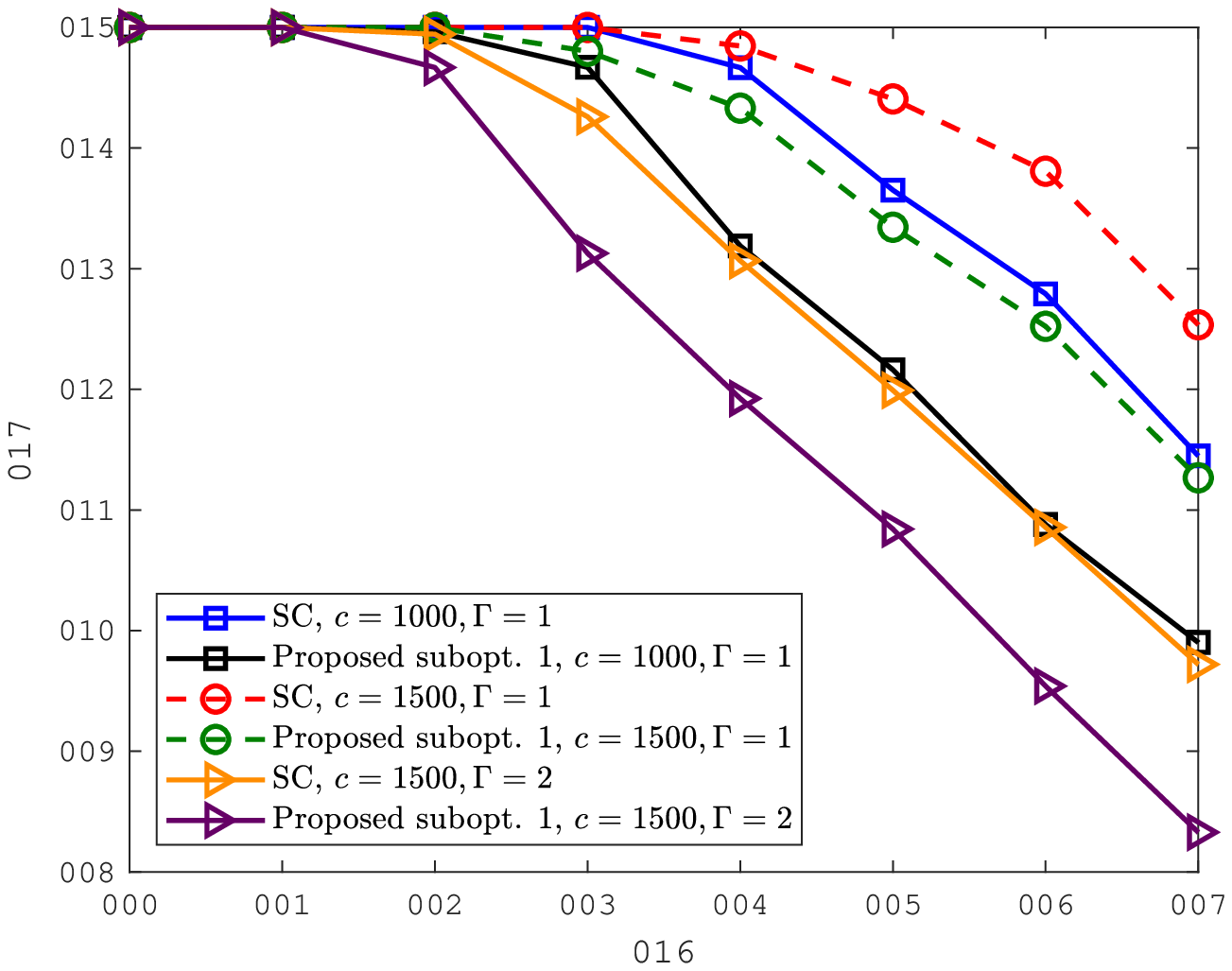}}\vspace{-4mm}
		\caption{Offloading probability vs. cell size in meters. $\tau =3,$ $K=4$, $D_{k}=7, \forall k$, $B_{k}=360~\text{bits}, \forall k$, and $r_{1}=20~\textrm{m}$.}
		\label{cell_raduispro}
	\end{minipage}\vspace{-0.75cm}
\end{figure*}
\begin{figure*}[!tbp]
	\centering		
		%%%%%%%%%%%%%%%%%%%%%%%%%%%%%%%%%%%%%%
%	\begin{minipage}{0.48\textwidth}
%		\hspace{-0.3cm}	
%		\resizebox{1\linewidth}{!}{\psfragfig{Codes/packeterror/fig/Pe}}\vspace{-4mm}
%		\caption{Average consumed power [dBm] vs. packet error probability, $r_{1}=r_{2}=75~\textrm{m}$, $\tau =3,$ $K=4$, $D_{1}=D_{2}=5$, and $D_{3}~=D_{4}~=7$, $B_{k}=160~\text{bits}, \forall k$, $c_{1}=c_{3}=330$, and $c_{2}~=c_{4}~=1500$.} 
%		\label{Pe}
%	\end{minipage}
		\begin{minipage}{0.48\textwidth}
	\hspace{-0.3cm}	
	\resizebox{1\linewidth}{!}{\psfragfig{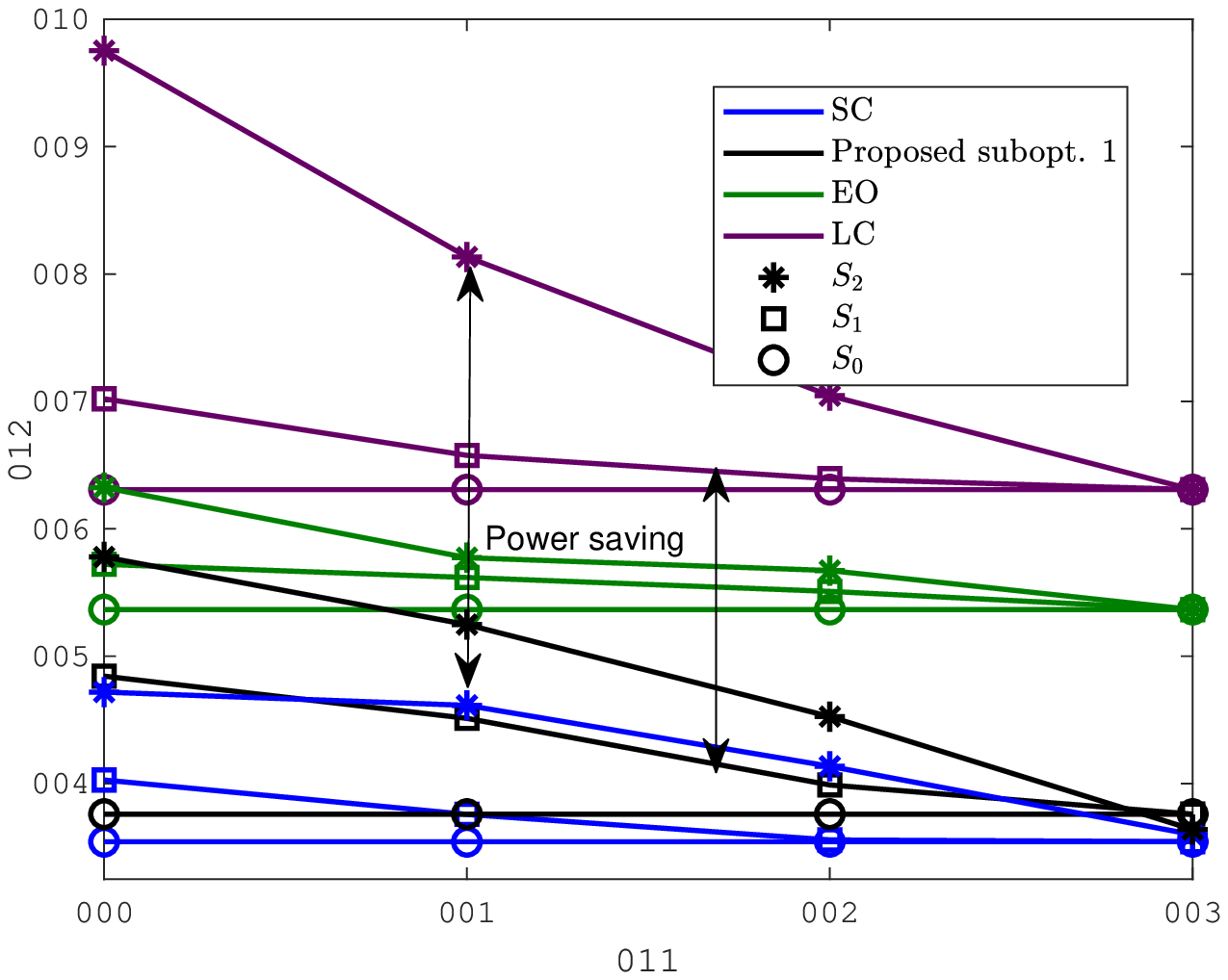}}\vspace{-4mm}
	\caption{Average consumed power [dBm] vs. delay in time slots and different delay scenarios. $r_{1}=r_{2}=75~\textrm{m}$, $\tau =2,$ $K=4$, $c_{1}=c_{3}=500$, and $c_{2}~=c_{4}~=1000$. $B_{k}=160~\text{bits}, \forall k$.} 
	\label{Delay}
\end{minipage}
\hfill
	\begin{minipage}{0.48\textwidth}
		\hspace{-0.3cm}	
	\resizebox{1\linewidth}{!}{\psfragfig{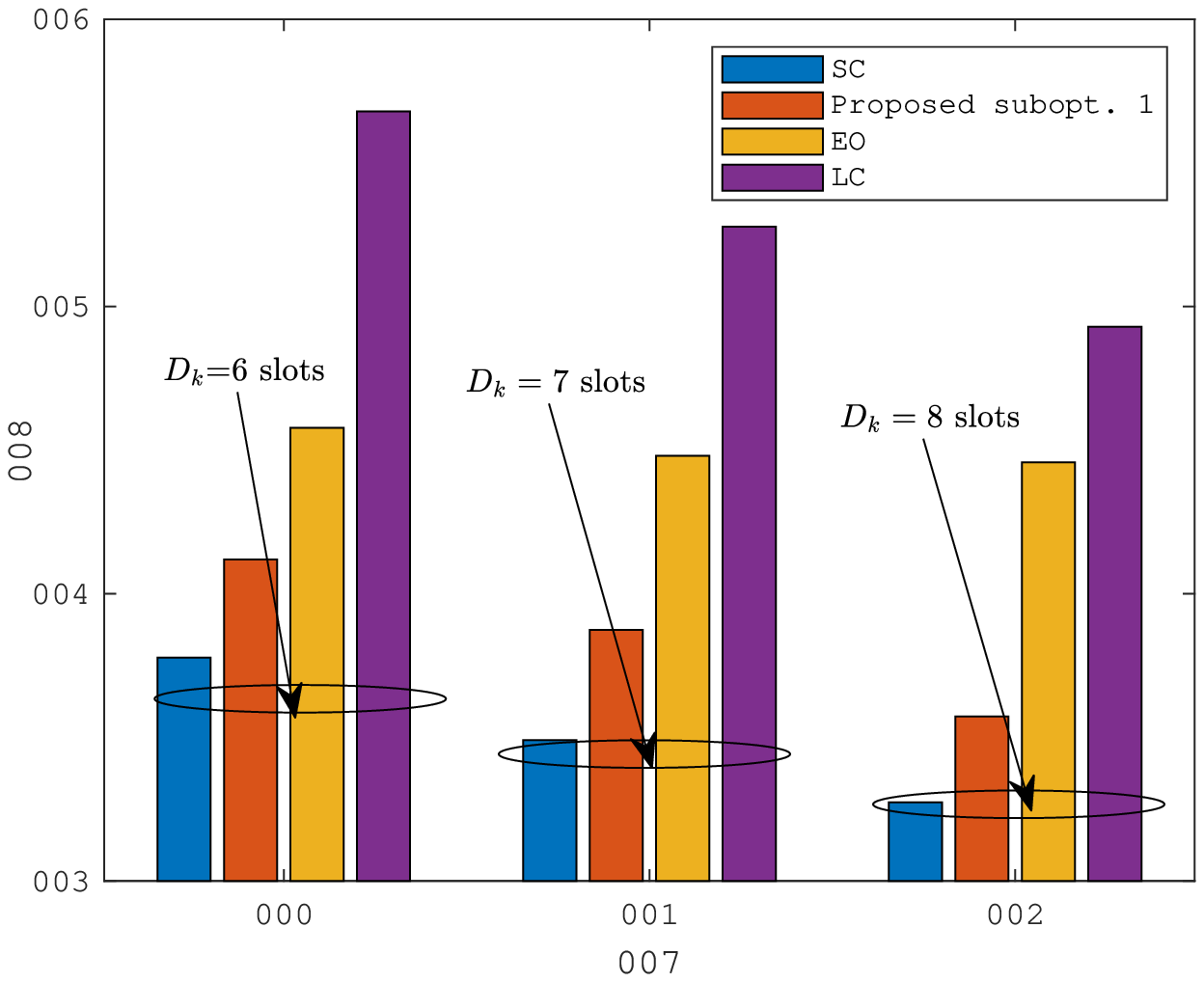}}\vspace{-4mm}
\caption{Impact of $\tau$ on the average system power consumption for different resource allocation schemes. $K=4$, $r_{1}=r_{2}=75~\textrm{m}$,  $D_{k}=\tau+N^{d}, \forall k$, $B_{k}=160~\text{bits}, \forall k$, $c_{k}=1500, \forall k$.}
		\label{tau}
	\end{minipage}
%\begin{minipage}[!tbp]
%	\centering	
%	\hspace{-0.3cm}	
%	\resizebox{0.44\linewidth}{!}{\psfragfig{Codes/Tau/tau}}\vspace{-4mm}
%	\caption{Impact of $\tau$ on the average system power consumption for different resource allocation schemes. $K=4$, $r_{1}=r_{2}=75~\textrm{m}$,  $D_{k}=\tau+N^{d}, \forall k$, $B_{k}=160~\text{bits}, \forall k$, $c_{k}=1500, \forall k$.}
%	\label{tau}
%\end{minipage}
	\vspace{-0.75cm}
\end{figure*}
  
In Fig.~\ref{cell_raduis}, we study the impact of the outer cell radius on the average system power consumption for different resource allocation schemes. As can be observed, increasing the outer cell radius increases the average system power consumption. This is due to the fact that the path loss increases with the distance, and as a result, more power is needed to maintain the same SNR for larger distances. For small outer radii, the performance of the proposed scheme is close to that of the EO scheme, as in this case, the proposed scheme is likely to offload the tasks of the users to the edge server because of the low transmission power needed. However, as the outer cell radius increases, the path loss increases, and thus the local users are more likely to compute the computation tasks locally to reduce power consumption. In this case, the performance of the proposed scheme
approaches that of the LC scheme. Fig.~\ref{cell_raduis} also shows the impact of $\Gamma=\Gamma_{k}, \forall k,$ on the system power consumption. As can be seen, the total system power consumption  is higher for larger $\Gamma$. This is due to the fact that as $\Gamma$ increases, the  size of the computation results to be transmitted in the downlink increases, and the BS has to allocate more power to satisfy the QoS constraint in the downlink.

In Fig.~\ref{cell_raduispro}, we investigate the impact of the outer cell radius on the offloading probability for the proposed low-complexity scheme 1 and SC for different values of $c=c_{k}, \forall k,$ and $\Gamma=\Gamma_{k}, \forall k$. As can be seen, increasing the outer cell radius reduces the probability of offloading. This is due to the fact that more power is needed to combat the path loss for larger distances, and thus, the users prefer to compute their tasks locally to reduce the total system power consumption. However, as the task complexity increases, i.e., for large numbers of required cycles $c$, the offloading probability increases. The reason for this behaviour is that as the number of cycles to process one bit increases, the CPU frequency must also increase to process the task within the required latency, and as a result, the local power consumption increases. Fig.~\ref{cell_raduispro} also reveals the impact of $\Gamma$ on the offloading probability. As can be seen, as $\Gamma$ increases, the offloading probability decreases. This is due to the fact that as $\Gamma$ increases, the size of the computed results in the downlink becomes larger, and the BS has to allocate more power to satisfy the QoS constraint in the downlink. In this case, the users are more likely to compute their tasks locally in order to limit the total system power consumption which leads to a lower offloading probability. 
   
In Fig.~\ref{Delay}, we investigate the effect of different delay requirements and consider three delay scenarios. For delay scenario $S_{0}$, all users have the same delay requirements, i.e., $D_{k}=6, \forall k$. For delay scenario $S_{1}$, we have $D_{1}=D,$ and $D_{k}=6, \forall k=\{2,3,4\}$. For delay scenario $S_{2}$, we have $D_{k}=D, \forall k=\{1,2,3\},$ and $D_{4}=6$. In Fig.~\ref{Delay}, we show the average system power consumption versus delay parameter $D$. As can be observed, the average system power consumption decreases with $D$, which is due to the fact that a larger $D$ increases the feasible set of problem (\ref{Op1}) and increases the flexibility of resource allocation. Moreover, the proposed suboptimal scheme attains large power savings compared to the LC scheme, especially, when the users have strict delay requirement. This is due to the limited computation capability of the users.

Fig.~\ref{tau} illustrates the impact of $\tau$ on the average system power consumption for different resource allocation schemes and $N^{u}=N^{d}=4$. As can be seen, the average system power consumption decreases as the value of $\tau$ increases.  The reason for this behaviour is that the number of overlapping time slots $\bar{O}=N^{u}-\tau$ is reduced as $\tau$ increases, and the feasible set of optimization problem (\ref{Op1}) become larger at the expense of an increase in the latency of the users, $D_{k}=\tau+n^{d}$. On the other hand, for small values of $\tau$, the number of overlapping time slots increases, and the total system latency is reduced for all users. This causes the average system power consumption to increase.    
\section{Conclusions}
This paper studied the resource allocation algorithm design for OFDMA-URLLC MEC systems. To ensure the stringent end-to-end transmission delay and reliability requirements of URLLC MEC systems, we proposed a joint uplink-downlink resource allocation scheme that takes into account FBT. Moreover, to minimize the end-to-end delay, we proposed a partial time overlap between the uplink and downlink frames which introduces a new
uplink-downlink causality constraint. The proposed resource
allocation algorithm design was formulated as an optimization problem for minimization of the total weighted transmit power
of the network under QoS constraints regarding the minimum
required number of computed bits of the URLLC users within
a maximum computation time. The resulting optimization problem was shown to be a non-convex mixed-integer problem and hard to solve. Nevertheless, we solved the optimization problem optimally using a branch-and-bound technique based on monotonic optimization theory. Moreover, to strike a balance between computation complexity and performance, we proposed two efficient suboptimal low-complexity schemes based on SCA. Our simulation results showed that the proposed resource allocation algorithm design facilitates the application of URLLC in MEC systems, and achieves significant power savings compared to several benchmark schemes. Moreover, our simulation results showed that the proposed suboptimal algorithms offer different trade-offs between performance and complexity and attained a close-to-optimal performance at comparatively low complexity.
	\section*{Appendix A}
The proof of Lemma \ref{lemmapena} follows similar steps as corresponding proofs in \cite{ghanem1, yan,kwan1}. In the following, we show that problems (\ref{Op2b}) and (\ref{Op3}) are equivalent. Let $U^{*}$ denote the optimal objective value of (\ref{Op3}). We define the Lagrangian function of problem (\ref{Op2b}), denoted by ${\mathcal {L}}(\boldsymbol{f},\bar{\mathbf {p}}^{u},\bar{\mathbf {p}}^{d},\mathbf {s}^{u},\mathbf {s}^{d},\boldsymbol{\alpha},\eta_{1},\eta_{2},\eta_{3})$, as \cite{Boyed}
\begin{align}&\hspace{-5cm}\label{La} {\mathcal {L}}(\boldsymbol{f},\bar{\mathbf {p}}^{u},\bar{\mathbf {p}}^{d},\mathbf {s}^{u},\mathbf {s}^{d},\boldsymbol{\alpha},\eta_{1},\eta_{2},\eta_{3})= \Phi(\boldsymbol{f}, \mathbf{\bar{p}}^{u}, \mathbf{\bar{p}}^{d},\boldsymbol{\alpha})+\eta_{1}(E^{u}(\mathbf{s}^{{u}})-H^{u}(\mathbf{s}^{{u}}))\nonumber\\&+\eta_{2}(E^{d}(\mathbf{s}^{{d}})-H^{d}(\mathbf{s}^{{d}}))+\eta_{3}(E^{\boldsymbol{\alpha}}(\boldsymbol{\alpha})-H^{\boldsymbol{\alpha}}(\boldsymbol{\alpha})),
\end{align}
where $\eta_{1}$, $\eta_{2}$, and $\eta_{3}$ are the Lagrange multipliers corresponding to constraints $\mbox {C6b}$, $\mbox {C8b}$ and $\mbox {C14b}$, respectively. 
Note that $E^{u}(\mathbf {s}^{u})-H^{u}(\mathbf {s}^{u}) \geq 0$, $E^{d}(\mathbf {s}^{d})-H^{d}(\mathbf {s}^{d}) \geq 0$, and $E^{\boldsymbol{\alpha}}(\boldsymbol{\alpha})-H^{\boldsymbol{\alpha}}(\boldsymbol{\alpha}) \geq 0$ hold.
Using Lagrange duality \cite{Boyed,Joinoptimization,ghanem1}, we have the following relation \footnote{Note that weak duality holds for convex and non-convex optimization problems\cite{Boyed}.}
\begin{align}\label{eq28} &U_{d}^{*}=\underset {\eta_{1}, \eta_{2}, \eta_{3} \ge 0}{ \mathop {\mathrm {max}}\nolimits } \quad \underset {\mathbf {p}^{u}, \mathbf {p}^{d},\mathbf {s}^{u},\mathbf {s}^{d}, \bar {\mathbf {p}}^{u},\bar{\mathbf {p}}^{d},\boldsymbol{f},\boldsymbol{\alpha} \in \boldsymbol{\Omega }}{ \mathop {\mathrm {min}}\nolimits } \quad {\mathcal {L}}(\boldsymbol{f},\bar{\mathbf {p}}^{u},\bar{\mathbf {p}}^{d},\mathbf {s}^{u},\mathbf {s}^{d},\boldsymbol{\alpha},\eta_{1},\eta_{2}, \eta_{3}) \IEEEyesnumber \IEEEyessubnumber
\\\overset {(a)}{\le }&\label{eq29}\underset {\mathbf {p}^{u}, \mathbf {p}^{d},\mathbf {s}^{u},\mathbf {s}^{d}, \bar {\mathbf {p}}^{u},\bar{\mathbf {p}}^{d},\boldsymbol{f},\boldsymbol{\alpha} \in \boldsymbol{\Omega }}{ \mathop {\mathrm {min}}\nolimits } \quad \underset {\eta_{1},\eta_{2},\eta_{3} \ge 0}{ \mathop {\mathrm {max}}\nolimits } \quad {\mathcal {L}}(\boldsymbol{f}, \bar{\mathbf {p}}^{u},\bar{\mathbf {p}}^{d},\mathbf {s}^{u},\mathbf {s}^{d},\boldsymbol{\alpha},\eta_{1},\eta_{2}, \eta_{3}) = U^{*},\IEEEyessubnumber
\end{align}
where $\boldsymbol{\Omega }$ is the feasible set specified by the constraints in (\ref{Op2b}).  
In the following, we first prove the strong duality, i.e., $U_{d}^{*}=U^{*}$. 
Let $( \mathbf {p}^{u*}, \mathbf {p}^{d*},\mathbf {s}^{u*},\mathbf {s}^{d*}, \bar {\mathbf {p}}^{u*},\bar{\mathbf {p}}^{d*},\boldsymbol{f}^{*},\boldsymbol{\alpha^{*}},\eta_{1}^{*},\eta_{2}^{*},\eta_{3}^{*})$ denotes the solution of (\ref{eq28}). For this solution, the following two cases are possible. \textit{Case 1)} If $E^{u}(\mathbf {s}^{u})-H^{u}(\mathbf {s}^{u})>0$, $E^{d}(\mathbf {s}^{d})-H^{d}(\mathbf {s}^{d})>0$, and $E^{\boldsymbol{\alpha}}(\boldsymbol{\alpha})-H^{\boldsymbol{\alpha}}(\boldsymbol{\alpha})>0$ hold, the optimal $\eta_{1}^{*}$, $\eta_{2}^{*}$, and $\eta_{3}^{*}$ are infinite, respectively. Hence, $U_{d}^{*}$ is infinite too, which contradicts the fact that it is upper bounded by a finite-value $U^{*}$.  \textit{Case 2)} If $E^{u}(\mathbf {s}^{u})-H^{u}(\mathbf {s}^{u})=0$, $E^{d}(\mathbf {s}^{d})-H^{d}(\mathbf {s}^{d})=0$, and $E^{\boldsymbol{\alpha}}(\boldsymbol{\alpha})-H^{\boldsymbol{\alpha}}(\boldsymbol{\alpha})=0$ holds, then $(\mathbf {p}^{u*}, \mathbf {p}^{d*},\mathbf {s}^{u*},\mathbf {s}^{d*}, \bar {\mathbf {p}}^{u*},\bar{\mathbf {p}}^{d*},\boldsymbol{f}^{*},\boldsymbol{\alpha^{*}})$ belongs to the feasible set of the original problem (\ref{Op2b}) which implies $U_{d}^{*}=U^{*}$. Hence, strong duality holds, and we can focus on solving the dual problem (\ref{eq28}) instead of the primal problem (\ref{eq29}). 

Next, we show that any $\eta_{1} \geq \eta_{1,0}$, $\eta_{2} \geq \eta_{2,0}$, and $\eta_{3} \geq \eta_{3,0}$ are optimal solutions for dual problem (\ref{eq28}), i.e., $\eta_{1}^{*}$, $\eta_{2}^{*}$ and $\eta_{3}^{*}$, where $\eta_{1,0}$, $\eta_{2,0}$, and $\eta_{3,0}$ are some sufficiently large numbers. To do so, we show that ${\Theta}(\eta_{1},\eta_{2},\eta_{3})\triangleq \underset {\bar{\mathbf {p}}^{u},\bar{\mathbf {p}}^{d},{\mathbf {p}}^{u},{\mathbf {p}}^{u},{\mathbf {s}}^{u},{\mathbf {s}}^{d}, \boldsymbol{f},\boldsymbol{\alpha}, \in \boldsymbol{\Omega }}{ \mathop {\mathrm {min}}\nolimits} \quad {\mathcal {L}}(\boldsymbol{f},\bar{\mathbf {p}}^{u},\bar{\mathbf {p}}^{d},\mathbf {s}^{u},\mathbf {s}^{d},\boldsymbol{\alpha},\eta_{1},\eta_{2},\eta_{3})$ is a monotonically increasing function of $\eta_{1}$, $\eta_{2}$, and $\eta_{3}$. Recall that $E^{u}(\mathbf {s}^{u})-H^{u}(\mathbf {s}^{u}) \geq 0$, $E^{d}(\mathbf {s}^{d})-H^{d}(\mathbf {s}^{d}) \geq 0$, and $E^{\boldsymbol{\alpha}}(\boldsymbol{\alpha})-H^{\boldsymbol{\alpha}}(\boldsymbol{\alpha}) \geq 0$ holds for any given ${\mathbf {p}}^{u},{\mathbf {p}}^{d},{\mathbf {s}}^{u},{\mathbf {s}}^{d},\bar{\mathbf {p}}^{u},\bar{\mathbf {p}}^{d}, \boldsymbol{f}, \boldsymbol{\alpha} \in \boldsymbol{\Omega }$.\newline Therefore, ${\mathcal {L}}(\boldsymbol{f}, \bar{\mathbf {p}}^{u},\bar{\mathbf {p}}^{d},\mathbf {s}^{u},\mathbf {s}^{d},\boldsymbol{\alpha},\eta_{1}(1),\eta_{2}(1),\eta_{3}(1))$ $ \leq {\mathcal {L}}(\boldsymbol{f},\bar{\mathbf {p}}^{u},\bar{\mathbf {p}}^{d},\mathbf {s}^{u},\mathbf {s}^{d},\boldsymbol{\alpha},\eta_{1}(2),\eta_{2}(2),\eta_{3}(2))$ holds for any given $\big(\bar{\mathbf {p}}^{u},\bar{\mathbf {p}}^{d},{\mathbf {p}}^{u},{\mathbf {p}}^{u},{\mathbf {s}}^{u},{\mathbf {s}}^{d},\boldsymbol{f},\boldsymbol{\alpha}\big)\in \boldsymbol{\Omega }$, $0 \leq \eta_{1}(1)\leq \eta_{1}(2)$, $0 \leq \eta_{2}(1)\leq \eta_{2}(2)$, and $0 \leq \eta_{3}(1)\leq \eta_{3}(2)$. This implies that ${\Theta}(\eta_{1}(1),\eta_{2}(1),\eta_{3}(1))$$\leq {\Theta}(\eta_{1}(2),\eta_{2}(2), \eta_{3}(2))$   and that ${\Theta}(\eta_{1},\eta_{2},\eta_{3})$ is monotonically increasing in $\eta_{1}$, $\eta_{2}$, and $\eta_{3}$. Using this result, we can conclude that  
$\Theta (\eta_{1},\eta_{2},\eta_{3}) =U^{*},\, \forall \eta_{1} \ge \eta_{1,0}$, $\eta_{2} \ge  {\eta}_{2,0}$, and $\eta_{3} \ge  {\eta}_{3,0}$.

In summary, due to strong duality, we can use the dual problem  (\ref{Op3}) to find the solution of the primal problem (\ref{Op2b}) and any $\eta_{1} \geq \eta_{1,0}$, $\eta_{2} \geq \eta_{2,0}$, and $\eta_{3} \geq \eta_{3,0}$ are optimal dual variables. These results are concisely given in Lemma \ref{lemmapena} which concludes the proof.
\bibliography{ref}  
\bibliographystyle{IEEEtran}
\end{document}